\documentclass[10pt,a4paper]{article}

\usepackage{amsmath,amssymb,amsthm,mathrsfs}
\usepackage{stmaryrd}  %
\usepackage{hyperref}
\usepackage{graphicx}
\usepackage{xcolor}
\usepackage{soul}
\usepackage{cancel}
\usepackage{fullpage}
\usepackage[inline]{enumitem}

\usepackage{titlesec}
\usepackage{cancel}

\theoremstyle{plain}
\newtheorem{theorem}{Theorem}[section]
\newtheorem{lemma}[theorem]{Lemma}
\newtheorem{proposition}[theorem]{Proposition}

\theoremstyle{definition}

\newtheorem{ass}[theorem]{Assumption}
\theoremstyle{remark}
\newtheorem{remark}[theorem]{Remark}
\numberwithin{equation}{section}

\newcommand{\bE}{\mathbb{E}}
\newcommand{\E}{\bE}

\newcommand{\bR}{\mathbb{R}}
\newcommand{\bN}{\mathbb{N}}

\newcommand{\etamax}{\Vert\eta\Vert_{\infty}}%

\newcommand{\etamin}{\Vert \eta^{-1}\Vert_\infty}%

\newcommand{\abs}[1]{\lvert#1\rvert} %

\DeclareMathOperator{\sign}{sign}

 \allowdisplaybreaks
\begin{document}
\title{Mean-Field Liquidation Games with Market Drop-out}

\author{Guanxing Fu\thanks{The Hong Kong Polytechnic University, Department of Applied Mathematics, Hung Hom, Kowloon, Hong Kong; guanxing.fu@polyu.edu.hk. G.
Fu’s research is supported by Hong Kong RGC (ECS) Grant No. 25215122 and NSFC Grant No. 12101523.} \qquad  Paul P. Hager\thanks{Humboldt University Berlin, Department of Mathematics, Unter den Linden 6,10099 Berlin; paul.hager@hu-berlin.de.} \qquad Ulrich Horst \thanks{Humboldt University Berlin, Department of Mathematics and School of Business and Economics, Unter den Linden 6,
10099 Berlin; horst@math.hu-berlin.de.}}

\maketitle

\begin{abstract}
We consider a novel class of portfolio liquidation games with market drop-out (``absorption''). More precisely, we consider mean-field and finite player liquidation games where a player drops out of the market when her position hits zero. In particular round-trips are not admissible. This can be viewed as a no statistical arbitrage condition. In a model with only sellers we prove that the absorption condition is equivalent to a short selling constraint. We prove that equilibria (both in the mean-field and the finite player game) are given as solutions to a non-linear higher-order integral equation with endogenous terminal condition. We prove the existence of a unique solution to the integral equation from which we obtain the existence of a unique equilibrium in the MFG and the existence of a unique equilibrium in the $N$-player game. We establish the convergence of the equilibria in the finite player games to the obtained mean-field equilibrium and illustrate the impact of the drop-out constraint on equilibrium trading rates.    
\end{abstract}



\section{Introduction}

	Models of optimal portfolio liquidation under market impact have received substantial consideration in the financial mathematics literature in recent years. Starting with the work of Almgren and Chriss \cite{AC-2001} existence and uniqueness of solutions in different settings have been established by a variety of authors including \cite{AJK-2014}, \cite{bank:voss:16}, \cite{FruthSchoenebornUrusov14}, \cite{GatheralSchied11}, \cite{GH-2017}, \cite{GHQ-2015}, \cite{GHS-2013}, \cite{HXZ-2020}, \cite{Kratz14}, \cite{KP-2016}, \cite{OW-2013}, and \cite{PZ-2018}. In this paper we consider a game-theoretic model of optimal portfolio liquidation where a player drops out of the market as soon as her position hits zero; both finite player games and the corresponding mean-field game are considered. 
	
	Mean-field games (MFGs) of optimal liquidation {\sl without} a strict liquidation constraint and {\sl without} market drop-out have been analyzed in the literature before. Cardaliaguet and Lehalle \cite{CL-2018} considered an MFG where each player has a different risk aversion. Casgrain and Jaimungal \cite{C-Jai-2018,C-Jai-2018b} considered games with partial information and different beliefs, respectively. Huang et al. \cite{HJN-2015} considered a game between a major agent who is liquidating a large number of shares and many minor agents that trade against the major player. Mean-field (type) games {\sl with} strict liquidation constraint have so far only been analyzed by Fu and co-workers \cite{FGHP-2018, FH-2020, FHX1} and Evangelista and Thamsten \cite{ET-2020}. 
	
	Our model is different, due to the market drop-out (``absorption'') constraint. In particular, we do not allow round-trips where players with zero (initial) position trade the asset to benefit from favorable future market dynamics. It has been shown in many papers including  \cite{FGHP-2018, FH-2020} that round-trips arise naturally in equilibrium when players with different initial holdings interact in the same market. At the same time beneficial round-trips are usually considered statistical arbitrage. Although the players do not generate positive profits in any state of the world, they generate profits on average from trading an asset that they do not hold. Our absorption constraint may hence be viewed as a ``no statistical arbitrage condition.'' %
		
		 Requiring a player to drop out of the market as soon as her position hits zero also avoids ``hot potato effects'' as they occur in \cite{Schied-2017b, Schied-2019} where different players repeatedly take long and short positions in the same asset to benefit from their own positive impact on market dynamics. It has been argued by many authors including ~\cite{Alfonsi2012, Gatheral2012} that cyclic fluctuations in the players' asset holdings should be viewed as model irregularities or a form of arbitrage and should hence be avoided. 
		
		In a benchmark model with only sellers we show that our drop-out constraint is equivalent to a short-selling constraint where the players are not allowed to change the direction of trading. {To the best of our knowledge, this is in accordance with legal US requirements that ban market participants (in US equity markets)} that liquidate portfolios on a client's behalf to take the opposite side of the market. In models with both buyers and sellers we prove that if sellers initially dominate, then the drop-out constraint is still equivalent to a no short-selling constraint on sellers, while buyers with sufficiently small initial conditions may still benefit from first selling the asset and then buying it back when the market environment is more favorable. We illustrate numerically that the drop-out condition leads to initially slower and eventually faster aggregate equilibrium trading rates compared to models without drop-out constraint. In the absence of a drop-out constraint some traders will initially ``oversell'' and buy the stock back later.  

	The assumption of absorption leads to endogenously controlled trading horizons. In stochastic settings games with endogenous trading horizons are challenging to analyze; the literature on such games is hence sparse. Specific stochastic MFGs with absorption have been considered by, e.g.~Campi and coauthors in \cite{campi2021, campi2018, campi2020}. Their arguments heavily rely on the non-degenerate assumption on the state dynamics, thus cannot be applied to our model because liquidation problems are degenerate. Graber and Bensoussan \cite{graber} studied Bertrand and Cournot MFGs with absorption, where the state was not allowed to be degenerate and thus it excludes the situation we are interested in. Graber and Sircar \cite{Sircar2023} studied a similar model by solving a master equation with Dirichlet boundary condition. Their arguments, too, rely on the non-degeneracy assumption on the state dynamics. Dumitrescu et al.~\cite{Dum}  developed a linear programming approach to study MFGs with stopping and absorption. They use relaxed solutions which allows them to prove existence results under weak assumptions and which easily lends itself to numerical implementation. Cesari and Zheng \cite{Zheng2022} established a necessary stochastic maximum principle for a class of stochastic control problems with absorption under strong assumptions that are difficult to verify in general. 
		
	The situation is simpler for deterministic MFGs. The work of Bonnans et al \cite{Bonnans}  establishes an abstract existence of solutions result for a class of finite-time MFGs of controls with mixed state-control and terminal state constraints. Their analysis is based on a sophisticated, yet abstract fixed point argument which makes it difficult to solve MFGs in closed form. A very different approach that applies to a different class of deterministic MFGs has recently been taken by Graewe et al. \cite{PUR}. They considered a deterministic linear-quadratic MFG of control in which players extract an exhaustible resource and drop out of the game as soon as their resources hit zero. In their model the representative player problem can be rephrased as a standard convex control problem with state constraints and unknwon terminal value of the adjoint process. The authors consider the Hamiltonian of the representative player model only up to a candidate optimal exploitation time in terms of which a candidate terminal condition for the adjoint process can be obtained. The key observation is that the dynamics of the amount of resource extracted by an individual player up to any {\sl given} time by using a strategy that would be optimal if the optimal exhaustion time was equal to that given time is independent of the {\sl equilibrium} mean production rate. 
			
	Our approach is inspired by the work of Graewe et al \cite{PUR}. However, unlike their model, our control problem {\sl cannot} be rewritten as a control problem with pointwise constraints on the controls and/or the state process. This makes it much more difficult to identify the Hamiltonian associated with the representative player's control problem. In the mean-field model we establish a verification result that states that given the optimal liquidation time, the ``usual'' Hamiltonian - and the forward-backward system derived from it by an application of Pontryagin's principle - where the absorption constraint is ignored is indeed the ``correct'' Hamiltonian - and hence giving us the correct forward-backward system on an endogenously determined time interval - of the representative player's problem if the exogenous mean trading rate does not change sign. This reduces the problem of solving the MFG to finding the optimal liquidation time for the representative agent problem. 
	
	While the sign condition on the mean trading rate is highly non-trivial and calls for a novel approach to solving the MFG, it does provide a clear path towards solving the game. Identifying Nash equilibria in the $N$-player game is much more difficult. Due to the endogenously determined trading horizons it seems impossible to derive a Nash equilibrium by solving a forward-backward system of Pontryagin-type. 
	
	Our key idea is to replace the mean-trading rate in the co-state equations (and only there) by an exogenous trading rate. This provides a unified mathematical framework within which to solve both the $N$-player game and the MFG. If $N=\infty$, then the system reduces to the forward-backward system for the representative player's problem in the MFG. If $N< \infty$, the system has no direct economic meaning but turns out to be very useful for establishing the existence of Nash equilibria in the $N$-player game. {We emphasize that our MFG provides the intuition for solving the $N$-player game. Instead of showing that equilibria in the MFG are approximate equilibria in the corresponding $N$-player game if the number of players is large we use the MFG to solve the $N$-player game. It is the nature of the MFG that motivates the forward-backward system in terms of which we characterize equilibria in the finite player game.} %
	
	 The solutions to the modified forward-backward system allow us to define a family of admissible trading strategies as functions of exogenous mean trading rates in terms of which we derive a general fixed-point equation for the candidate Nash equilibria in both games. The fixed-point equation can be rewritten in terms of a complex integral equation that  seems hard to solve at first sight. However, motivated by the previously established verification result for the MFG we expect any solution to the said equation to be of constant sign. This motivates an a priori estimate from which we do indeed deduce this property. The ``constant sign property'' substantially simplifies the integral equation and eventually allows us to prove an existence and uniqueness of solutions result.  
	
	Subsequently, we prove a necessary and sufficient maximum principle that states that the solution to our integral equations yields the unique mean-field equilibrium in a subclass of equilibria with continuous mean trading rates. The analysis of the $N$-player game is more subtle as we do not have a general verification result. Instead, we prove that a verification result applies ``in equilibrium'' from which we deduce the existence of a unique Nash equilibrium with continuous aggregate trading rate in the finite player game. 	 In a final step, we prove the convergence of the unique equilibria in the $N$-player game to the mean-field equilibrium established before when the number of players tends to infinity.  
	
	The reminder of this paper is organized as follows. In Section \ref{sec:single-player} we introduce our liquidation models along with a family of benchmark trading strategies in terms of which we can formulate our fixed-point condition for equilibrium aggregate trading rates. The fixed-point equation is analyzed in Section \ref{sec:MFG}. Existence, uniqueness of convergence of equilibrium results are established in Section 4. Section 5 illustrates the impact of drop-out constraints on equilibrium trading. Our numerical simulations suggest that the MFG equilibrium is a very good approximation of finite-player equilibria even for relatively small numbers of players. Section 6 concludes.

{\sl Notation.} For $y\in\mathbb R^{d_1\times d_2}$, $\|y\|$ denotes its $2$-norm. For $y\in L^\infty([0,T]), C^{0}([0,T])$ we denote by $\|y\|_\infty$ the supremum norm. For an almost everywhere differentiable function $y$, we denote by $\dot y$ its derivative. For space $\mathcal S$ of real-valued functions we denote by $\mathcal S_\pm$ the subspace of $\mathcal S$ whose elements take nonnegative/nonpositive values.

\section{The Model}\label{sec:single-player}

In this section we introduce a game-theoretic liquidation model with permanent price impact where a player drops out of the game as soon as her portfolio process hits zero. 

\subsection{Single player model with market drop-out}\label{sec:2}

We consider the problem of a large investor that needs to liquidate a large number $x\in\bR$ of shares over the time interval $[0,T]$. Following the majority of the optimal liquidation literature we assume that only absolutely continuous trading strategies are allowed. In particular, the portfolio process satisfies
$$X_t = x-\int_0^t \xi_s \,ds, \quad t \in [0,T]$$
where $\xi_t$ denotes the trading rate at time $t \in [0,T]$; positive rates mean that the trader is selling the asset while negative rates mean she is buying it. We assume that the player drops out of the market as soon as her position hits zero. As a result, the set of admissible trading strategies is given by the set 
\begin{equation*}
	\mathcal A_{x}:=\left\{\xi\in L^2([0,T]) \;\bigg|\;  \exists  \tau\in(0, T] \text{ s.t. } \abs{X_t} > 0 \mbox{ for } t \in [0,\tau) \mbox{ and } X_t = 0 \mbox{ for } t \in [\tau,T] \right\}
\end{equation*}
of all square integrable processes $\xi$ that satisfy almost surely the liquidation constraint $X_T = 0$ and whose associated portfolio process is absorbed at zero.

We also assume that the unaffected price process against which the trading costs are benchmarked follows some Brownian martingale $S$, and that the trader’s transaction price process is given by
\[
\tilde S_t = S_t - \int_0^t \kappa_s \xi_s\,ds - {\frac{1}{2}} \eta_t \xi_t, \quad t \in [0,T]
\]
for deterministic positive impact processes $\kappa$ and $\eta$. The integral term accounts for permanent price impact while the term ${\frac{1}{2}}\eta_t \xi_t$ accounts for the instantaneous price impact that does not affect future transactions. {The liquidation cost $C$ is then defined as the difference between the book value and the proceeds from trading as
\[
	C = xS_0 - \int_0^T \tilde S_t \xi_t\, dt.
\]
Doing integration by parts and then taking expectations the martingale terms drops out and the expected liquidation cost equals 
\[
	\mathbb{E}[C] = \int_0^T \left( \frac{1}{2}\eta_t\xi^2_t+\kappa_t \xi_t X_t \right) \, dt.
\]}
Introducing an additional risk term  ${\frac{1}{2}}\lambda_t X_t^2$ for some deterministic non-negative process $\lambda$ that penalizes slow liquidation, the trader’s optimization problem reads 
\begin{equation}
\min_{\xi \in {\mathcal{A}}_x} J(\xi) \quad \mbox{ s.t. } \quad dX_t = - \xi_t\, dt, 
\end{equation}
where the cost function is given by
\[
	J(\xi):=\int_0^T \left(\frac{1}{2}\eta_t\xi^2_t+\kappa_t \xi_t X_t+\frac{1}{2}\lambda_t X^2_t\right)\,dt.
\]

\subsection{Game-theoretic liquidation models with market drop-out}

In $N$-player games of optimal liquidation it is usually assumed that transaction prices are driven by the average trading rate 
\begin{align*}
\overline\xi^N_t ~:=~ \frac{1}{N}\sum_{i=1}^{N} \xi^k_t 
\end{align*}
and that the transaction price process for some player $i=1, ..., N$ is of the form 
\[
	\tilde S^i_t = S_t - \int_0^t \kappa_s \overline\xi^N_s ds - \eta_t \xi^i_t.
\]
Assuming that all players are homogeneous, the cost functional for a generic player $i$ equals
\[
	J(\xi^i; \xi^{-i}):=\int_0^T \left(\frac{1}{2}\eta_t(\xi^i_t)^2+\frac{\kappa_t X^i_t}{ N}  \sum_{j=1}^N \xi^j_t  +\frac{1}{2}\lambda_t (X^i_t)^2\right)\,dt
\]
and his optimization problem reads 
\begin{equation}
\min_{\xi \in {\mathcal{A}}_{x_i}} J(\xi^i;\xi^{-i}) \quad \mbox{ s.t. } \quad dX^i_t = - \xi^i_t dt,\quad X_0^i=x_i.
\end{equation}

{
\begin{remark}
Some comments on our cost function are in order. Following \cite{FGHP-2018} we assume that aggregate net trading rates drive asset prices (mid-quote prices in our setting). The assumption that permanent market impact depends on aggregate behavior is standard in the literature on liquidation games, see e.g. \cite{CL-2018,C-Jai-2018b,FGHP-2018}; 
it accounts for the fact that (i) prices are driven by excess rather than absolute demand and (ii) the impact on an individual player on (mid-quote) prices  should be lower in markets with many participants. 
By contrast we assume that the instantaneous impact depends on individual, not aggregate demand. 
The reason for this is threefold. First, buyers and sellers act on different sides of the market. 
It would hence be inappropriate to assume that  buying and selling effects average out when computing instantaneous impact. 
Second, different sellers never trade at exactly the same time.\footnote{Trading rates should be viewed as approximating situations where market participants arrive at independent Poisson arrival times with rates $\xi^i$ in which case no two traders arrive at exactly the same time.}
Third, assuming that the instantaneous impact scales in the number of sellers would be equivalent to assuming that the instantaneous impact parameter $\eta$ converges to zero as the number of participants converges to infinity. The recent work of Horst and Kivman \cite{HK-2023} suggests that the limiting optimization problem would be very different from the pre-limit one.  In any case, the scaling assumptions are irrelevant in finite player games.  
\end{remark}
}

In the corresponding MFG the average trading rate is replaced by an exogenous trading rate $\mu$, the representative player's cost functional is given by  
\[
	J(\xi; \mu):=\int_0^T \left(\frac{1}{2}\eta_t\xi^2_t+\kappa_t \mu_t X_t+\frac{1}{2}\lambda_t X^2_t\right)\,dt
\]
and her control problem reads %
\begin{equation}\label{opt-rp}
	\min_{\xi \in \mathcal A_x} J(\xi;\mu) \quad \mbox{ s.t. } \quad dX_t = - \xi_t dt,\quad X_0=x.
\end{equation}
Given an optimal trading rate {$\xi^{*,x}(\mu)$ for the representative player with initial position $x$} as a function of the exogenous mean trading rate $\mu$ the equilibrium condition reads
{
$$\mu = \int_{\mathbb{R}}\xi^{\ast, x}(\mu)\nu_0(dx),$$
where $\nu_0$ denotes the distribution of the initial position.} We impose the following standing assumption on the cost coefficients. 

\begin{ass}\label{ass:single-player} The cost coefficients satisfy\footnote{The differentiability of $\kappa$ is only needed to solve the $N$-player game. It is not needed to solve the MFG.} $$ \lambda \in L^{\infty}_+([0,T]), ~\kappa\in C^{1}_+([0,T]), \quad \mbox{ and } \quad 1/\eta, \eta\in C^{1}([0,T];(0,\infty)).$$
\end{ass}

\subsection{Heuristics}

In this section we {\it heuristically} derive a family of strategies in terms of which we can formulate a fixed-point problem for equilibrium mean trading rates; {a rigorous proof is deferred to Section \ref{sec:single_player_verification}}. Following Graewe et al \cite{PUR} the main idea in solving our liquidation games is to reduce the problem of finding the optimal liquidation strategies to the problem of finding the optimal liquidation time for which we derive an explicit representation. 

One of the main difficulties when solving our liquidation games with market drop-out is that the underlying control problems are control problems of absorption that {\sl cannot} be rewritten as control problems with pointwise constraints on the controls and/or state process. {Rewriting the problem as a problem with state constraints would require an additional non-negativity constraint on the control as in  \cite{PUR}. } 

{In literature, control problems with absorption are documented as {\it exit time control problems}. As far as we can tell, there are basically three ways to address such problems. First, one could use the standard Hamiltonian and derive an HJB equation with a zero boundary condition $u(t,0) \equiv 0$ that captures the absorption constraint. In the presence of boundary conditions it is common practice to assume bounded domains and/or compact control sets; see e.g. \cite{FS-2006}. Assuming bounded domains would not be appropriate in our setting; moreover, in our setting the terminal condition is singular $\lim_{t\nearrow T}u(t,x) = +\infty {{\bf 1}}_{\{x\neq 0\}}$, hence the boundary condition on the parabolic boundary $\{0 \}\times[0,T]\cup\mathbb R\times\{T\}$ is not smooth. The second approach is based on weak formulation of control problems, e.g. \cite{campi2018}, where the absorption time is ``exogenous''; it is the absorption time for the canonical process. Thus, the Hamiltonian is the standard one. However, the weak formulation approach requires nondegeneracy of state dynamics, which is not our case. Third, one could apply a stochastic maximum principle as in \cite{Zheng2022} in which case one needs to take the optimality of the absorption time into account. In other words, one needs to differentiate the absorption time w.r.t.~the strategy. This would result in additional terms in the Hamiltonian and hence in a highly non-standard adjoint equation.  }

\subsubsection{Candidate strategy}

{In this paper we follow a different approach introduced in \cite{PUR} and consider the standard adjoint equation but on an endogenously determined time interval. Our educated guess is that  in equilibrium - hence for the optimal absorption time\footnote{The absorption time in our problem means the first liquidation time.} $\tau^i$ - the Hamiltonian 
\[    
	H(t,\xi^i,X^i,Y^i; \xi^{-i} )=-\xi^i Y^{i}  + \frac{1}{2}\eta_t(\xi^i)^2+\kappa_t \overline\xi^N_t X^i_t +\frac{1}{2}\lambda_t(X^i_t)^2
\]
resulting from an application of the standard stochastic maximum principle without absorption constraint} is the correct Hamiltonian for player $i$'s optimization problem in the $N$-player game {on the time interval $[0,\tau^i]$.\footnote{Intuitively, once an optimal liquidation time has been identified, there is no need to differentiate it.}} Assuming that this is in fact true, taking partial derivatives with respect to $X^i$ and $\xi^i$, {and still taking the liquidation constraint into account,} we expect to obtain the candidate best response strategy 
\begin{equation}\label{eq:xi-i}
	\xi^i_t=\frac{Y^{i}_t-\frac{\kappa_t}{N}X^i_t}{\eta_t}
\end{equation}
in terms of the solution to a system of $N$ fully coupled systems of forward-backward differential equations
\begin{equation}\label{eq:XY-i}
	\left\{\begin{split}
		\bigg.  \dot X^i_t=&~-\frac{Y^i_t-\frac{\kappa_t}{N}X^i_t   }{\eta_t} 1_{\{t \le \tau^{i}\}}  \\
		\bigg.  -\dot Y^i_t=&~(\lambda_t X^i_t+\kappa_t\overline{\xi}^{N}_t  )1_{\{t \le \tau^{i}\}}  \\
		X_0^i=&~x_i,\quad X^{i}_T = 0
	\end{split}\right., \qquad \text{for a.e. } t\in[0,T]. %
\end{equation}

Solving these systems simultaneously for all players is challenging, due to the dependence on the endogenous absorption times. Our idea is to consider instead - for any $\delta \in [0,1]$, any initial position $x \in \bR$, aggregate trading rate $\mu$ and candidate absorption time $\tau_\mu(x)$ - the auxiliary forward-backward system 
\begin{equation}\label{eq:FBODE} %
	\left\{\begin{split}
		\bigg. \dot X_t=&~- \frac{Y_t-\delta \kappa_t X_t   }{\eta_t}  1_{\{t \le \tau_\mu(x)\}}\\
		\bigg.  -\dot Y_t=&~	(\lambda_t X_t+\kappa_t\mu_t)1_{\{t \le \tau_\mu(x)\}} \\
		X_0=&~x,\quad  X_T = 0
	\end{split}\right., \qquad \text{for a.e. } t\in[0,T]. %
\end{equation}
For $\delta=\frac{1}{N}$ this system corresponds to the system \eqref{eq:XY-i} where the aggregate trading rate $\overline\xi^N$ in the co-state equation is replaced by a generic trading rate $\mu$. For $\delta=0$ the system reduces to the corresponding system for the MFG where $N=\infty$ and the term $ \frac{\kappa_t}{N}X_t$ drops out of the equation. 

{ 
\begin{remark}
The above system can be viewed as a forward-backward ODE with an initial and terminal condition on the forward process $X$ and an unknown terminal condition on the backward process $Y$.
\end{remark}
}

We will see that the above system provides a unified mathematical framework within which to analyze both the $N$-player game and the MFG but there are some subtleties to be accounted for. Specifically, we will {derive a particular solution $(X^{\delta,x,\mu}, Y^{\delta,x,\mu})$ to the above system for a {\sl specific} absorption time $\tau_\mu(x)$ -  that is to be determined - from which we then define the candidate strategy}
\begin{equation}\label{xi*}
	\xi^{\delta, x, \mu} := \frac{Y^{\delta,x,\mu} - \delta \kappa X^{\delta,x,\mu}}{\eta}.
\end{equation} 
The challenge is to characterize the absorption time; this will be achieved heuristically in the Section \ref{sec:candidate-absorption-time}. Subsequently, in Section \ref{sec:single_player_verification}, we prove that the candidate is indeed the absorption time of $\xi^{\delta, x, \mu}$ if the process $\mu$ does not change sign. From this we will deduce that the candidate strategy is indeed admissible and indeed an optimal response to $\mu$. 

{For a given distribution $\nu_0$ of the players' initial positions, we then proceed by showing that the fixed-point problem for the aggregated candidate strategy}
\begin{equation} \label{mu*}
	\mu = \int_{\mathbb{R}} \xi^{\delta, x, \mu} \nu_0(dx)
\end{equation} 
{has a unique solution for all $\delta \in [0,1]$ and, importantly, that this solution does not change its sign.
This finally allows us to verify that this fixed point yields an equilibrium aggregate trading rate in the MFG ($\delta = 0$). The verification for the $N$-player game ($\delta = \frac 1 N$) is slightly different and will be carried out in Section 4.}

\subsubsection{Candidate absorption time}\label{sec:candidate-absorption-time}

{In this section we heuristically identify a candidate optimal absorption time and, hence, a candidate optimal liquidation strategy. Motivated by liquidation problems without market drop-out as analyzed in, e.g.~\cite{FGHP-2018} our idea is to determine a solution to \eqref{eq:FBODE} by making the linear ansatz $Y=AX+B$ where the coefficients $A$ and $B$ satisfy the following ODE system:

\begin{equation}\label{eq:AB}
	\left\{\begin{split}
		-\dot A^\delta_t=&~-\frac{(A^\delta)^2_t}{\eta_t} + \delta \frac{\kappa_t}{\eta_t} A^\delta_t +\lambda_t\\
		-\dot B^\delta_t=&~\left(-\frac{A^\delta_tB^\delta_t}{\eta_t}+\kappa_t\mu_t\right) 1_{\{t\leq\tau_\mu(x)\}}\\
		\lim_{t\nearrow T}A^\delta_t=&~\infty,\quad B^\delta_{\tau_\mu(x)} =0.
	\end{split}\right. 
\end{equation}

\begin{remark}
For models without market drop-out where $\tau_\mu(x)=T$ it is well known that solving the above ODE system with singular terminal condition is equivalent to solving the forward-backward system \eqref{eq:FBODE} with unknown terminal condition, and that the unique optimal strategy is given by \eqref{xi*}. Our idea is to follow a similar approach. We will see that we can solve the system \eqref{eq:AB} for any $\tau_\mu(x) \in [0,T]$. In particular, for any such time we can {\sl define} the backward (adjoint) process $Y=AX+B$ and then solve for the forward (portfolio) process $X$ in  \eqref{eq:FBODE} with initial condition $X_0=x$. However, we cannot expect the terminal condition $X_{\tau_\mu(x)}=0$ to hold and hence the strategy in \eqref{xi*} to be admissible. For the strategy to be admissible, we will first identify those times $\tau_\mu(x)$ for which the liquidation constraint holds and then identify additional conditions on $\mu$ such that $\tau_{\mu(x)}$ is the {\sl first} liquidation time for the process $X$.  
\end{remark}
}

We notice that the dynamics of the process $A^\delta$ is independent of $\tau_\mu(x)$. The existence of a unique process $A^{\delta}$ that satisfies the singular terminal value problem follows from Lemma~\ref{lem:A-kappa}. The second equality in \eqref{eq:AB} then yields that
\begin{equation}\label{eq:explicit-B}
	B^\delta_t = B^{\delta,x,\mu}_t=1_{ \{t\leq\tau_\mu(x)\} }\int_t^{\tau_\mu(x)} e^{ -\int_t^s\frac{A^{\delta}_r}{\eta_r}\,dr	}\kappa_s\mu_s\,ds ,\qquad t\in [0,T].
\end{equation}

{
Armed with the solution $(A^\delta,B^{\delta,x,\mu})$ to the above singular ODE system we can now plug the linear ansatz 
	$Y^{\delta,x,\mu}:=A^\delta X+B^{\delta,x,\mu}$}
along with \eqref{eq:explicit-B} into \eqref{eq:FBODE} to obtain the following representation of the forward process $X^{\delta,x,\mu}$ in  \eqref{eq:FBODE}:
\begin{align}\label{eq:explicit-X}
  X^{\delta,x,\mu}_t=\left\{	xe^{-\int_0^t \frac{A^{\delta}_r - \delta\kappa_r}{\eta_r}\,dr}	-\int_0^t \frac{1}{\eta_s} e^{-\int_s^t \frac{A^{\delta}_r- \delta\kappa_r}{\eta_r}\,dr }\int_s^{\tau_\mu(x)} \kappa_u \mu_u e^{-\int_s^u\frac{A^{\delta}_r}{\eta_r}\,dr}\,du\,ds	\right\}1_{ \{  t\leq \tau_\mu(x) \} }.
\end{align}

{We emphasize that the above process - albeit being well defined for arbitrary $\tau_\mu(x) \in [0,T]$ - may not be an admissible portfolio process in general; see also Lemma \ref{lem:a_priori_candiate} below.} However, our educated guess is that the {optimal} absorption time $\tau_\mu(x)$ is implicitly given by\footnote{We will see that this is true if $\mu$ does not change sign. }  
\begin{equation}\label{cand-time2}
	   X_{\tau_\mu(x)}^{\delta, x, \mu } = 0.
\end{equation}
We proceed by
 identifying those times $\tau_\mu(x)$ that satisfy \eqref{cand-time2}. We start with the following standard result. {It states that in models without drop-out constraint $X^{\delta,x,\mu}$ is an admissible portfolio process.} The proof follows from \eqref{eq:explicit-X} together with Lemma \ref{lem:estimatesA} and Lemma \ref{lem:A-kappa}. 

\begin{proposition}\label{prop:liquidation}
	If $\tau_\mu(x) = T$, then $X^{\delta,x,\mu}_{\tau_\mu(x)} =0$.
\end{proposition} 

Let us now turn to the case where {we expect early liquidation to take place}. In this case, the liquidation time satisfies $\tau_\mu(x)< T$.
From \eqref{cand-time2} it must hold that
\begin{equation}\label{eq:=x}
	\int_0^{\tau_\mu(x)}\frac{1}{\eta_s} e^{\int_0^s\frac{A^{\delta}_r-\delta\kappa_r}{\eta_r}\,dr} \int_s^{\tau_\mu(x)}\kappa_u\mu_u e^{-\int_s^u\frac{A^{\delta}_r}{\eta_r}\,dr}\,du\,ds=x.  %
\end{equation}
To identify those initial positions for which early liquidation {may take place} we introduce the function
\[
h^{\delta}_t:=e^{	-\int_0^t\frac{A^{\delta}_r}{\eta_r}\,dr	}\int_0^t \frac{1}{\eta_s} e^{ \int_0^s\frac{2A^{\delta}_r - \delta\kappa_r}{\eta_r}\,dr}\,ds.
\]
As we will see, the function $h^\delta$ features in the characterization of the absorption time. 
The proof of the following lemma is given in the appendix. 

\begin{lemma}\label{lem:about-h} 
The function $h^{\delta}$ is strictly increasing, differentiable, and both $h^{\delta},$ $\dot{h}^{\delta}$ and $1/\dot{h}^{\delta}$ are bounded uniformly in $\delta\in[0,1]$. Moreover, 
	\begin{equation}\label{def:alpha}
		\alpha^{\delta}_T := \left(\lim_{t\to T}h^\delta_t\right)^{-1} = \lim_{t \to T} A^{\delta}_t e^{-\int_0^t\frac{A^{\delta}_r - \delta\kappa_r}{\eta_r}\,dr} \in (0,\infty),
	\end{equation}
and for almost every $t\in[0,T]$ it holds that $h_t^\delta\rightarrow h_t^0$ and $\dot h_t^\delta\rightarrow\dot h_t^0$, as $\delta\rightarrow 0$. 	 
\end{lemma}

We can apply Fubini's theorem to rewrite the left hand side of the equation \eqref{eq:=x} as a function of time
\begin{equation}
	\label{cha-tau}
	\begin{split}
		& \int_0^t\frac{1}{\eta_s}e^{\int_0^s\frac{A^{\delta}_r-\delta\kappa_r}{\eta_r}\,dr} \int_s^t \kappa_u\mu_u e^{	-\int_s^u\frac{A^{\delta}_r}{\eta_r}\,dr	}\,du\,ds \\
		=&~\int_0^t\kappa_u\mu_u e^{	-\int_0^u\frac{A^{\delta}_r}{\eta_r}\,dr	}\int_0^u \frac{1}{\eta_s} e^{\int_0^s\frac{2A^{\delta}_r - \delta\kappa_r}{\eta_r}\,dr}\,ds\,du \\
		=& ~  \int_0^t\kappa_u\mu_uh^{\delta}_u\,du  \\
		=: & ~ f_\mu(t).
	\end{split}
\end{equation}
In view of the preceding lemma the function $f_\mu$ is well defined. 
{Using the convention $\inf \emptyset = T$ we see that $\tau_\mu(x)$ is a candidate first liquidation time only if   
\begin{equation}\label{eqn:tau}
	\tau_\mu(x) = \inf\left\{ t \in [0,T] : f_\mu(t) = x\right\}.
\end{equation}

To summarize, our heuristic analysis suggests to proceed as follows:
	\begin{itemize}
		\item Define the function $f_\mu$ given in \eqref{cha-tau} and the candidate absorption time given by \eqref{eqn:tau}.
		\item For the candidate absorption time define the functions $(A^\delta,B^{\delta,x,\mu})$ by \eqref{eq:AB} and \eqref{eq:explicit-B}.
		\item In terms of the functions $(A^\delta,B^{\delta,x,\mu})$ define the portfolio process $X^{\delta,x,\mu}$ by \eqref{eq:explicit-X} and the ``adjoint'' process $Y^{\delta,x,\mu}=A^\delta X^{\delta,x,\mu}+B^{\delta,x,\mu}$. 
		\item Define the candidate best response by \eqref{xi*}.
	\end{itemize}       
In the next section, we verify rigorously that $(X^{\delta,x,\mu},Y^{\delta,x,\mu})$ satisfies \eqref{eq:FBODE}, that $\xi^{\delta,x,\mu}$ is the best response if $\mu$ does not change sign, and that $\tau_\mu(x)$ defined by \eqref{eqn:tau} is the indeed desired absorption time. }

{
\begin{remark}\label{rmk:sign}
We emphasize that \eqref{eqn:tau} is a necessary, yet a priori no sufficient condition for $\tau_\mu(x)$ to be the optimal liquidation time.\footnote{We thank an anonymous referee for bringing this subtle point to our attention.} We will see that the condition is sufficient only if the aggregate strategy $\mu$ does not change sign. In this case, important quantitative result can be inferred from the above charcaterization of the optimal liquidation time. Let us assume that the sign of $\mu$ is always positive while the sign of $x$ is negative, or vice versa. In this case, $\tau_\mu(x) = T$.  That is, if the aggregate trading rate is always positive (negative), then we expect early liquidation for a buyer (seller) not to be beneficial in equilibrium as the price process is driven in a favorable direction. Moreover in this case, a seller, respectively, buyer will liquidate early only if 
\[
	x \in \left(0,f_\mu(T) \right), \quad \mbox{respectively,} \quad 
	x \in \left(f_\mu(T), 0 \right).
\]
That is, we obtain explicit bounds beyond which early liquidation is not beneficial.
\end{remark}}

\subsection{Admissibility and verification in the mean-field model}\label{sec:single_player_verification}

In this section we substantiate the heuristic considerations from above. In a first step we prove that if $\mu$ does not change sign, and if $\tau_\mu(x)$ is defined by \eqref{eqn:tau}, then the candidate strategy $\xi^{\delta, x, \mu}$ defined by \eqref{xi*} is admissible and that $\tau_\mu(x)$ is the associated {\sl first} liquidation time. 

\begin{lemma}\label{lem:a_priori_candiate}
	\begin{itemize}
		\item[i)] Let $\mu \in L^{1}([0,T])$,  $x \in \bR$. Let {$\tau_\mu(x)$ be defined by \eqref{eqn:tau}, let $ X^{\delta,x, \mu}$ be defined by $\eqref{eq:explicit-X}$ and let
		$$Y^{\delta, x, \mu} := A^{\delta}   X^{\delta, x, \mu} + B^{\delta, x, \mu},$$
		where $(A^{\delta}, B^{\delta, x,\mu})$ is the unique solution to \eqref{eq:AB}.
		Then $(X^{\delta, x, \mu},  Y^{\delta, x, \mu} )$ is absolutely continuous and a solution to the forward-backward ODE \eqref{eq:FBODE}. 	}
		\item[ {ii)}]  {The candidate strategy $ {\xi^{\delta, x, \mu}}$ defined in \eqref{xi*} is absolutely continuous on $[0,T]$} and there exists a constant $C>0$ that depends only on $\mu, \eta, \lambda, \kappa$ and $T$ such that
		\begin{align}\label{eq:xi_apriori_estiate}
			\Vert \xi^{\delta, x, \mu}\Vert_{\infty} + \Vert \dot\xi^{\delta, x, \mu}\Vert_\infty \le C(1+\abs{x}), \qquad x\in \bR, \; \delta\in[0,1].
		\end{align}
		\item[{iii)}]{ If $\mu\in L^{1}_-([0,T])\cup L^{1}_+([0,T])$, then the candidate strategies are admissible, i.e., $$\xi^{\delta, x, \mu} \in \mathcal{A}_x$$ for all $x \in \bR$ and $\delta\in[0,1]$. In particular, $\tau_\mu(x)$ defined in \eqref{eqn:tau} is the corresponding first liquidation time.}
	\end{itemize}
\end{lemma}
\begin{proof}
	\begin{itemize}
		\item[i)] {It is straightforward to verify that the process $(X^{\delta, x, \mu}, Y^{\delta, x, \mu})$ is 
		almost everywhere differentiable and that it satisfies the desired ODE on $[0,\tau_\mu(x))$. If $\tau_\mu(x)=T$, then the assertion follows from Proposition~\ref{prop:liquidation} along with the fact that  
		\begin{equation}\label{limY}
		\lim_{t\to T} Y^{\delta, x, \mu}_t = \lim_{t\nearrow T} (A^{\delta}_tX^{\delta, x, \mu}_t+B_t^{\delta,x,\mu}) = \alpha^{\delta}_T(x-f_\mu(T))< \infty
		\end{equation}
		where $\alpha_T^\delta$ was defined in \eqref{def:alpha}. 
		If $\tau_\mu(x) < T$, then   $X^{\delta, x, \mu}_{\tau_\mu(x)} = 0$ by the definition of $\tau_\mu(x)$. Thus, $X^{\delta, x, \mu}\vert_{[\tau_\mu(x),T]} = 0$. Since $B^\delta_{t} = 0$ for all $t \in [\tau_\mu(x),T]$ it follows that $Y^{\delta, x, \mu}\vert_{[\tau_\mu(x),T]} = 0$. In particular, $X^{\delta, x, \mu}$ satisfies the desired boundary conditions and $(X^{\delta,x, \mu}, Y^{\delta,x, \mu})$ solves the forward-backward system \eqref{eq:FBODE} on the whole interval $[0,T]$.  
		}
	
		\item[ii)]  {Absolute continuity of $\xi^{\delta, x, \mu}$ follows i).} Its boundedness can be seen as follows.  
		 		For any $t\in[0,T]$, using Lemma \ref{lem:convergence-alpha} in the second and Lemma \ref{lem:about-h} in the third step, there exists a constant $c>0$ such that
		\begin{align*}
			 \abs{\xi^{\delta, x, \mu}_t} \le&~ \frac{1}{\eta_t}\left(A^{\delta}_t - \delta\kappa_t \right) e^{-\int_0^t\frac{A^{\delta}_r - \delta \kappa_r}{\eta_r}\,dr}\left\{	\abs{x}+\int_0^T \frac{1}{\eta_s}e^{\int_0^s\frac{A^{\delta}_r - \delta \kappa_r}{\eta_r}\,dr} \int_s^T \kappa_u\abs{\mu_u} e^{-\int_s^u\frac{A^{\delta}_r}{\eta_r}\,dr}\,du\,ds 	\right\}
			\\
			&+ \frac{1}{\eta_t}\int_0^T e^{ -\int_t^s\frac{A^{\delta}_r}{\eta_r}\,dr}\kappa_s\abs{\mu_s}\,ds \\
			\le&~ \left(A_0^\delta+\|\kappa\|_\infty\right)\left\|1/\eta\right\|_\infty \left\{ \abs{x} + \int_0^{T}h^{\delta}_t \kappa_t \abs{\mu_t}\,dt\right\}+ \left\|1/\eta\right\|_\infty\int_0^T\kappa_s\abs{\mu_s}\,ds \\
			\le&~ \left(A_0^\delta+\|\kappa\|_\infty\right)\left\|1/\eta\right\|_\infty  \left\{ \abs{x} +\Vert h^{\delta}\Vert_\infty\Vert\kappa\Vert_\infty \Vert \mu\Vert_{L^{1}([0,T])}\right\}+ \left\|1/\eta\right\|_\infty\Vert\kappa\Vert_\infty \Vert \mu\Vert_{L^{1}([0,T])}\\
			\le&~ c(1+\abs{x}).
		\end{align*}
		Furthermore, for any $t \in [0, \tau_\mu(x))$ where $\xi^{\delta, x, \mu}$ is differentiable we have from \eqref{eq:FBODE} that
		\begin{align*}
			\abs{\dot\xi^{\delta, x, \mu}_t} 
			&\le~ \left|{\frac{\dot\eta_t}{\eta_t}}\right| | \xi^{\delta, x, \mu}_t| + \frac{1}{\eta_t} \abs{\dot Y^{\delta, x, \mu}_t  -\delta\dot\kappa_t X_t^{*,x,\mu}-\delta\kappa_t\dot X_t^{*,x,\mu} }  \\
			&\le~  {\Vert\dot\eta\Vert_\infty}\|1/\eta\|_{\infty}   \Vert \xi^{\delta, x, \mu}\Vert_\infty 
			\\&~\quad+ \|1/\eta\|_\infty \Big(\Vert\kappa\Vert_\infty\Vert\mu\Vert_\infty  + (\Vert\lambda\Vert_\infty +\|\dot\kappa\|_\infty)\Vert X^{\delta, x, \mu}  \Vert_\infty   + \|\kappa\|_\infty \|\xi^{\delta,x,\mu}\|_\infty \Big).
		\end{align*}
		Since $\Vert X^{\delta, x, \mu} \Vert_\infty \le \abs{x} + T \Vert \xi^{\delta, x, \mu}\Vert_\infty \le (cT+1)\abs{x} + cT$, we conclude that  
		\begin{align*}
			\abs{\dot\xi^{\delta, x, \mu}_t} \le   C (1+\abs{x}),
		\end{align*}
		for some constant $C>c$. Finally, since $\xi^{\delta, x, \mu}$ is almost everywhere differentiable and constant on the interval $(\tau_{\mu}(x), T]$ the stated estimate follows.
	{	
		\item[iii)] 
		To prove that $\xi^{\delta, x, \mu}$ is admissible it remains to show that $\tau_\mu(x)$ is the first time $X^{\delta, x, \mu}$ hits zero. Since $\mu$ dos not change sign by assumption  we can distinguish the following three cases. 
		\begin{itemize}
			\item The case $x = 0$ is trivial, since $\tau_\mu(0) = 0$.	
			\item 
		If $x$ has a different sign from $\mu$, then $\tau_\mu(x)=T$ and it follows from \eqref{eq:explicit-X} that 
		$$
		|X^{\delta, x, \mu}_t| \ge |x|  e^{-\int_0^t \frac{A^{\delta}_r - \delta\kappa_r}{\eta_r}\,dr} > 0, \qquad t \in[0,T).
		$$
		This implies that $\tau_\mu(x)=T$ is the indeed first liquidation time.
		\item If $x$ have the same sign as $\mu$ and if $\tau_\mu(x)<T$, then it follows from \eqref{eq:explicit-X} and the definition of $\tau_\mu(x)$ that for $t\leq \tau_\mu(x)$
		\begin{equation*}
			\begin{split}
	X^{\delta, x, \mu}_t=&~ e^{-\int_0^t \frac{A^{\delta}_r - \delta\kappa_r}{\eta_r}\,dr}\left\{	x	-\int_0^t \frac{1}{\eta_s} e^{\int_0^s \frac{A^{\delta}_r- \delta\kappa_r}{\eta_r}\,dr }\int_s^{\tau_\mu(x)} \kappa_u \mu_u e^{-\int_s^u\frac{A^{\delta}_r}{\eta_r}\,dr}\,du\,ds	\right\} \\
	=&~ e^{-\int_0^t \frac{A^{\delta}_r - \delta\kappa_r}{\eta_r}\,dr}\left\{	x- f_\mu(\tau_\mu(x))+ \int_{t}^{\tau_\mu(x)}\frac{1}{\eta_s} e^{\int_0^s\frac{A^{\delta}_r-\delta\kappa_r}{\eta_r}\,dr} \int_s^{\tau_\mu(x)}\kappa_u \mu_u e^{-\int_s^u\frac{A^{\delta}_r}{\eta_r}\,dr}\,du\,ds	\right\}.
		\end{split}	
		\end{equation*}
Let us now assume to the contrary that $X^{\delta, x,\mu}_t = 0$ for some $t < \tau_\mu(x)$. 
Using that $$f_\mu(\tau_\mu(x)) = x$$ we see from the above representation and the sign assumption on $\mu$ that $\kappa_u \mu_u = 0$ on $[t, \tau_\mu(x)]$.
This implies that $$f_\mu(t) = f_\mu((\tau_\mu(x))) =x$$ which contradicts the definition of $\tau_\mu$. 

If $\tau_\mu(x) = T$, then by the definition of $\tau_\mu(x)$ it holds that  %
	\begin{equation*}
		\begin{split}
			|X^{\delta, x, \mu}_t|=&~ e^{-\int_0^t \frac{A^{\delta}_r - \delta\kappa_r}{\eta_r}\,dr}\left|	x	-\int_0^t \frac{1}{\eta_s} e^{\int_0^s \frac{A^{\delta}_r- \delta\kappa_r}{\eta_r}\,dr }\int_s^{\tau_\mu(x)} \kappa_u \mu_u e^{-\int_s^u\frac{A^{\delta}_r}{\eta_r}\,dr}\,du\,ds	\right| >0,\quad t<T.\\
		\end{split}	
\end{equation*}
Then $T$ is the first time at which $X^{\delta,x,\mu}$ hits zero.
\end{itemize}
} 
	\end{itemize}    
\end{proof}

For the MFG ($\delta = 0$) the system \eqref{eq:FBODE} reduces to the forward-backward system that results from an application of Pontryagin's maximum principle to the candidate Hamiltonian. The following proposition verifies that {\sl if the aggregate trading rate does not change sign} the admissible strategy $\xi^{0, x, \mu}$ solves the representative player's problem.

The verification result is key to our equilibrium analysis for at least two reasons. First, it motivates the fixed-point equation \eqref{mu*}. Second, we prove in Section 3 that the fixed-point equation reduces to a complex integral equation that seems difficult to solve. From the verification result we expect the solutions to the equation to not change sign, which motivates an important a priori estimate that substantially simplifies the analysis that follows.   
 
\begin{proposition}\label{prop:verification}
If $\mu\in L^{1}_-([0,T]) \cup L^{1}_+([0,T])$ then the strategy defined in \eqref{xi*} with $\delta = 0$ is the a.s.~unique optimal strategy that solves the representative player's optimization problem  \eqref{opt-rp}.
\end{proposition}
\begin{proof}%
	Most of the arguments given below are a blend of arguments given in \cite{FGHP-2018} and \cite{PUR}. We present a detailed proof to keep the paper self-contained.  

Let us now assume that $\mu\in L^{1}_-([0,T]) \cup L^{1}_+([0,T])$ {and let $x\in\bR$. Further, let $\xi^{0, x, \mu}$ be the candidate strategy defined in \eqref{xi*}, which by Lemma~\ref{lem:a_priori_candiate}.iii) is admissible} and let $\xi { \in \mathcal{A}_x}$ be any other admissible strategy with corresponding state process $X$ and absorption time $\tau$. We need to distinguish two cases, depending on which strategy liquidates first.
		
		\begin{itemize}
			\item {\textbf{$\tau\leq \tau_{\mu}(x)$}}. Note that $\xi_t=X_t=0$ on $(\tau,  T]$ and thus we have
			\begin{equation*}
				\begin{split}
					& ~ J(\xi;\mu)-J(\xi^{0, x, \mu};\mu) \\
					=&~\int_0^{\tau_{\mu}(x)}\left(	\frac{1}{2}\eta_t\xi^2_t+\kappa_t\mu_t X_t+\frac{1}{2}\lambda_t X^2_t \right)\,dt \\
					& \quad - \int_0^{\tau_{\mu}(x)}\left(\frac{1}{2}\eta_t(\xi^{0, x, \mu}_t)^2+\kappa_t\mu_t X^{0, x, \mu}_t+\frac{1}{2}\lambda_t (X^{0, x, \mu}_t)^2			\right)\,dt\\
					\geq&~\int_0^{\tau_{\mu}(x)} \eta_t\xi^{0, x, \mu}_t(
					\xi_t-\xi^{0, x, \mu}_t) + (\kappa_t\mu_t+\lambda_t X^{0, x, \mu}_t)(X_t-X^{0, x, \mu}_t)\,dt.
				\end{split}
			\end{equation*}
			Since $Y^{0, x, \mu}$ is continuous, integration by parts along with equation \eqref{eq:FBODE} yields that
			\begin{equation}
				\label{YX}
				\begin{split}
					0 = & ~ Y^{0, x, \mu}_{\tau_{\mu}(x)}(X_{\tau_{\mu}(x)} -X^{0, x, \mu}_{
						\tau_{\mu}(x)}	) \\
					= & ~ -\int_0^{\tau_{\mu}(x)}Y^{0, x, \mu}_t(\xi_t-\xi^{0, x, \mu}_t)\,dt-\int_0^{\tau_{\mu}(x)}(X_t-X^{0, x, \mu}_t)(\kappa_t\mu_t+\lambda_t X^{0, x, \mu}_t)\,dt.
				\end{split}
			\end{equation}
			Plugging this into the difference of cost functionals we see that
			\[
			J(\xi;\mu)-J(\xi^{0, x, \mu};\mu) \ge \int_0^{\tau_{\mu}(x)} (\eta_t\xi^{0, x, \mu}_t-Y^{0, x, \mu}_t)(X_t-X^{0, x, \mu}_t)\,dt=0,
			\]
			where the last equality follows from the definition of $\xi^{0, x, \mu}$.
			
			\item $\tau > \tau_\mu(x)$. In this case $\tau_\mu(x) < T$ and since $\mu$ does not change sign we must have that $\sign(\mu_0) = \sign(x)$. Since $\tau$ is the absorption time of $X$ at zero this shows that 
			\[
			\mu_t X_t \geq 0 \quad \mbox{for all} \quad t \in [0,T].
			\]
			As a result, the same arguments as before show that 
			\begin{equation*}
				\begin{split}
					& ~ J(\xi;\mu) -J(\xi^{0, x, \mu};\mu) \\
					=&~\int_0^{\tau_{\mu}(x)}\left(	\frac{1}{2}\eta_t\xi^2_t+\kappa_t\mu_t X_t+\frac{1}{2}\lambda_t X^2_t\right)\,dt \\
					& \quad - \int_0^{\tau_{\mu}(x)}\left(	\frac{1}{2}\eta_t(\xi^{0, x, \mu}_t)^2+\kappa_t\mu_t X^{0, x, \mu}_t+\frac{1}{2}\lambda_t (X^{0, x, \mu}_t)^2			\right)\,dt\\
					& \quad \quad +	\int_{\tau_{\mu}(x)}^\tau\left(	\frac{1}{2}\eta_t\xi^2_t+\kappa_t\mu_t X_t+\frac{1}{2}\lambda_t X^2_t\right)\,dt	\\
					\geq &~0.
				\end{split}
			\end{equation*}
		\end{itemize}
		Since the above inequalities are strict if $\Vert \xi^{0, x, \mu} - \xi\Vert_{L^{2}([0,T])} \neq 0$, it follows that $\xi^{0, x, \mu}$ is the unique optimizer. From this we conclude that \eqref{opt-rp} admits a unique solution.
\end{proof}

\begin{remark}[Monotonicity of trading]
The explicit representation of the optimal trading strategy allows us to deduce an important implication for trading. If $\mu$ does not change sign, then a representative seller always sells if $\mu>0$ and a representative buyer always buys if $\mu<0$. At the same time, if $\mu > 0$, then for any buyer the {non-negative} process $B^{0, \mu, \tau_\mu(x)} = B^{0, \mu, T}$ is independent of $x$. For small enough negative initial position {and non-trivial}\footnote{ { For buyers, Remark \ref{rmk:sign} implies that $\tau_\mu(x)=T$. Then \eqref{eq:explicit-B} implies that $B_0^{\delta,x,\mu} =\int_0^Te^{-\int_t^s \frac{A^\delta_r}{\eta_r}\,dr  }\kappa_s\mu_s\,ds$. Thus, $B^{0,x,\mu}_0=0$ if and only if $\kappa\equiv 0$ a.e., which is a trivial case with only one player. In this trivial case, it is easy to see that trading is monotone even for buyers.} } {$\kappa$} this means that $$\xi^{0, x, \mu}_0=\frac{A^{0}_0x+B^{0,\mu,T}_0}{\eta} > 0.$$ Hence, if sellers dominate, then for a buyer with small enough initial position, it is beneficial to initially sell the asset to benefit from the positive price drift when buying it back. 
\end{remark}

\section{The equilibrium equation}\label{sec:MFG}

In what follows we denote by $\xi^{\delta, \mu}$ the vector of strategies introduced in \eqref{xi*} and introduce the mapping   
\[
	F: L^1([0,T]) \to \mathbb{R}^{[0,T]}, \qquad  \mu \mapsto \int_{\bR}\xi^{\delta, x, \mu}\nu_0(dx).  
\]

The verification result given in Section 2 suggests that any fixed-point of the mapping $F$ for $\delta =0$ that does not change sign yields an MFG equilibrium. A verification result for a generic aggregate trading rate $\mu$ cannot be expected in the $N$-player game. However, we {expect} a verification result to hold for equilibrium mean trading rates, that is, for fixed-points of the map $F$. This suggests that both games can be solved within a common mathematical framework by following four non-standard steps:\footnote{For MFGs the fourth step reduces to verifying that $\mu^*$ does not change sign.}
\begin{align}\label{equlibrium_framwork}
\renewcommand{\arraystretch}{1.5}
\left\{
	\begin{array}{l}
		1. \text{ Fix } \mu\in L^{1}([0,T]).\\
		2. \text{ Consider the candidate strategy profile $\xi^{\delta,\mu}$ obtained in \eqref{xi*}  for $\delta=0$, resp. $\delta=\frac 1 N$}. \\
		3. \text{ Find the fixed-points $\mu^*$ of the mapping  $\mu \mapsto F(\mu)$ in $L^{1}([0,T])$}.\\
		4. \text{ Verify that $\xi^{\delta,\mu^*}$ is a Nash equilibrium.} %
	\end{array}
\right.
\end{align}

In what follows,  we prove that for any fixed $\delta\in[0,1]$ there exists a fixed point 
\begin{align}\label{eq:consistency}
	\mu^*_t &=~  F(\mu^*)_t , \qquad t \in [0, T]
\end{align}
that does not change sign and thereby we address the third step of \eqref{equlibrium_framwork} for both the MFG and the $N$-player game.
To guarantee that the fixed-point mapping is always well defined we impose the following assumption on the initial distribution $\nu_0$ of the players' portfolios and the market impact parameters.\footnote{The assumption on the parameters is purely technical. It is satisfied if, for instance, the permanent impact parameter is constant.} 
\begin{ass}\label{ass:mean-field} The distribution of initial position $\nu_0$ has a finite first absolute moment and
\[
	~\lambda + \delta \dot\kappa \geq  0.
\]
\end{ass}

\subsection{A non-linear integral equation}

We begin our fixed-point analysis by deriving a more explicit form of the fixed-point map that will later allow us to rewrite the fixed point property in terms of a non-linear integral equation. {In what follows we denote by 
\begin{equation}\label{barf}
	\overline f_\mu(t):=\max_{0\leq s\leq t} f_\mu(s)\quad\textrm{and}\quad \underline f_\mu(t):=\min_{0\leq s\leq t}f_\mu(s)
\end{equation}
the running maximum and the running minimum of the function $f_\mu$, respectively. }

\begin{lemma}\label{lem:mu_differentiability}
For any $\mu\in L^{1}([0,T])$ it holds for all $t \in [0,T]$ that
\begin{align} \label{eq:consitency2}
	F(\mu)_t  ~= F(\mu)_T + \int_t^{T}\frac{1}{\eta_s}\int_{I_\mu(s)}\Big(\kappa_s\mu_s+(\lambda_s+\delta\dot\kappa_s) X^{\delta, x, \mu}_s + (\dot\eta_s-\delta\kappa_s)\xi^{\delta,x,\mu}_s\Big) \nu_0(dx)\,ds,
\end{align}
where $I_\mu(t) := (-\infty, \underline{f}_\mu(t)]\cup[\overline{f}_\mu(t), \infty)$. In particular,   $F$ maps $L^{1}([0,T])$ into the space of absolutely continuous functions on $[0,T]$.
\end{lemma}
\begin{proof}
From the definition of $\xi^{\delta, x, \mu}$ and $\tau_\mu(x)$ in Section 2 it holds that
\[
	\xi^{\delta,x,\mu}\big\vert_{(\tau_\mu(x), T]}\equiv0,\quad\textrm{ for all } x\textrm{ such that }\tau_\mu(x)<T.
\]
Furthermore, by Lemma \ref{lem:a_priori_candiate} it follows that $\xi^{\delta, x, \mu}$ is almost everywhere differentiable and the derivative has at most linear growth in $x$ uniformly in $t$.
Therefore, by Assumption~\ref{ass:mean-field} on $\nu_0$ we can apply Fubini's theorem to the right-hand side of equation \eqref{eq:consistency} to obtain by the definition of $\xi^{\delta,x,\mu}$ and $\tau_\mu(x)$
\begin{align*}
\int_{\bR}\xi^{\delta, x, \mu}_t\nu_0(dx) 
&=~ \int_{\bR}\xi^{\delta, x, \mu}_T\nu_0(dx)  - \int_t^{T}\int_{I_\mu(s)}\dot \xi^{\delta, x, \mu}_s\nu_0(dx)  \; ds \\
&=~F(\mu)_T - \int_t^{T}\int_{I_\mu(s)}\frac{d}{ds}\left(\frac{Y^{\delta, x, \mu}_s - \delta \kappa_s X_s^{\delta, x,\mu}}{\eta_s}\right)\nu_0(dx)  \; ds, \qquad t\in [0,T].
\end{align*}
The assertion now follows by using that $(X^{\delta, x, \mu}, Y^{\delta, x, \mu})$ solves the forward-backward equation~\eqref{eq:FBODE}.
\end{proof}

In what follows we denote the left and right tail distribution function of $\nu_0$ by
\begin{align*}
p_0:\bR\to[0,1],& \quad x\mapsto\nu_0((-\infty, x]), \\
q_0:\bR\to[0,1],& \quad x \mapsto\nu_0([x,\infty)).
\end{align*}
In other words, $q_0$ represents the tail distribution function for sellers and $p_0$ the tail distribution function for buyers. 
Furthermore, we introduce the mean position among all players
$$\E[\nu_0] := \int_{-\infty}^{\infty}x\nu_0(dx).$$

We now derive a representation of a solution $\mu$ to our fixed-point equation in terms of an integral equation.
Similar to the proof of Lemma~\ref{lem:mu_differentiability} and using that $X^{\delta, x, \mu}_T = 0$ we have that
\begin{align*}
\int_{I_\mu(t)}X^{\delta, x, \mu}_t \nu_0(dx) = \int_{I_\mu(t)} \int_{t}^{T} \xi^{\delta, x, \mu}_s \, ds\, \nu_0(dx) =  \int_{t}^{T} \int_{\bR} \xi^{\delta, x, \mu}_s \nu_0(dx)\, ds =\int_t^{T}F(\mu)_s ds.
\end{align*}

Using the definition of the tail probabilities $q_0$ and $p_0$ equation \eqref{eq:consitency2} can be rewritten as 
\begin{equation} \label{eq:consitency3}
\begin{split}
	F(\mu)_t  =&~ F(\mu)_T + \int_{t}^{T}\frac{\kappa_s}{\eta_s}\left(q_0\big( \overline{f}_\mu(s)\big) + p_0\big(\underline{f}_\mu(s) \big) \right)\mu_s\, ds
		+ \int_t^{T}\frac{\lambda_s + \delta\dot\kappa_s}{\eta_s}\left(\int_s^{T}F(\mu)_u du\right)ds \\
		&~+ \int_t^{T}\frac{\dot\eta_s - \delta \kappa_s}{\eta_s}F(\mu)_s\;ds, \qquad t\in[0,T].
\end{split}
\end{equation}

Recalling the definition of $\overline{f}_\mu$ and $\underline{f}_\mu$ in \eqref{barf}, we see that any fixed point $\mu = F(\mu)$ solves the non-linear higher order integral equation \eqref{eq:consitency3}. The following proposition shows that converse is also true. Any solution to this equation yields a fixed point $\mu = F(\mu)$. 

\begin{proposition}\label{prop:fix-point-integral} A process $\mu\in L^{1}([0,T])$ solves the fixed-point equation \eqref{eq:consistency} if and only if $\mu_T = F(\mu)_T$ and $\mu$ solves the equation
\begin{equation} \label{eq:mu-first-integral}
\begin{split}
	\mu_t  =&~ \mu_T +\int_{t}^{T}\frac{\kappa_s}{\eta_s}\left(q_0\big( \overline{f}_\mu(s)\big) + p_0\big(\underline{f}_\mu(s) \big) \right)\mu_s\, ds
		+\int_t^{T}\frac{\lambda_s + \delta\dot\kappa_s}{\eta_s}\left(\int_s^{T}\mu_u du\right)ds \\
		&~+ \int_t^{T}\frac{\dot\eta_s - \delta \kappa_s}{\eta_s}\mu_s\;ds, \qquad t\in[0,T].
\end{split}
\end{equation}
\end{proposition}
\begin{proof} We have already proven that any solution to the fixed point equation \eqref{eq:consistency} satisfies the integral equation \eqref{eq:mu-first-integral}. 
For the converse direction, define the linear operator $G$ on $L^{1}([0,T])$ by
\begin{align*}
G(\mu)_t = \int_t^{T}\frac{\lambda_s + \delta\dot\kappa_s}{\eta_s}\left(\int_s^{T}\mu_u du\right)ds + \int_t^{T}\frac{\dot\eta_s - \delta \kappa_s}{\eta_s}\mu_s\;ds, \qquad t\in [0,T].
\end{align*}
Now let $\mu\in L^{1}([0,T])$ be a solution to \eqref{eq:mu-first-integral}.
Comparing the equations \eqref{eq:consitency3} and \eqref{eq:mu-first-integral} we then obtain
\begin{align*}
\mu - F(\mu)  &=~  G(\mu) - G(F(\mu))\\
&=~ G(\mu - F(\mu)).
\end{align*}
An application of Gr\"onwall's inequality \cite[Theorem 2.7]{Teschl-2016} shows that any solution to the linear integral equation $\varphi = G(\varphi)$ is necessarily trivial and hence $\mu = F(\mu)$.
\end{proof}

The integral equation \eqref{eq:mu-first-integral} is rather intricate as it depends on the running maximum and minimum of the function $f_\mu$ and hence depends in a non-linear way on the entire process $\mu$. Furthermore, the terminal condition $\mu_T = F(\mu)_T$ is endogenous and needs to be characterized more explicitly.

The following lemma provides an a priori estimate that allows us to substantially simplify the equation.

\begin{lemma}\label{lem:mu-sign}
Let $\mu \in L^{1}([0,T])$ be a solution to the equation \eqref{eq:mu-first-integral}. Then, 
\begin{equation}\label{eq:mu-sign}
e^{-K_1(T-t)}\frac{\eta_T}{\eta_t} \vert \mu_T\vert\le \abs{\mu_t}\leq \abs{\mu_T} e^{K_2(T-t)}, \qquad t \in [0,T],
\end{equation}
where $K_1:= \delta\Vert\kappa\Vert_\infty  \Vert 1 /\eta\Vert_\infty$ and $K_2  := \Vert 1 / \eta \Vert_\infty((1+\delta)\Vert \kappa \Vert_{\infty} + \delta T \Vert \dot\kappa \Vert_{\infty} + T\Vert \lambda \Vert_{\infty} + \Vert\dot{\eta}\Vert_\infty)$.
In particular,
$$\sign(\mu_t) = \sign(\mu_T), \qquad t\in[0,T].$$
\end{lemma}
\begin{proof}
It follows from the equation \eqref{eq:mu-first-integral} that for all $t\in [0,T]$,
\begin{align*}
		\abs{\mu_t} \leq&~ \abs{\mu_{T}} + \Vert 1 /\eta\Vert_\infty \Vert \kappa \Vert_{\infty}\int_t^T\abs{\mu_s}\,ds +  \Vert 1 /\eta\Vert_\infty(\Vert \lambda \Vert_{\infty} + \delta\Vert \dot\kappa \Vert_{\infty})\int_t^T\int_s^T \abs{\mu_r} \,dr\,ds \\
		&+  \Vert 1 /\eta\Vert_\infty (\Vert\dot{\eta}\Vert_\infty + \delta \Vert \kappa \Vert_\infty)\int_t^T\abs{\mu_s}\,ds\\
		\leq &~ \abs{\mu_{T}} +K_2\int_t^T \abs{\mu_s} \,ds.
\end{align*}
Hence, the upper estimate in \eqref{eq:mu-sign} follows from Gr\"onwall's inequality.
In particular $\mu \equiv 0$ for $\mu_T=0$. Let us thus suppose that $\mu_T > 0$ and put $$t_0 := \inf\{t\in[0,T)\;|\;\mu_{s} > 0, \; t \le s \le T \}.$$
By Assumption~\ref{ass:mean-field} it holds that $(\lambda + \delta\dot\kappa)$ is non-negative.\footnote{For the MFG this assumption is always satisfied.}
Since $\mu$ is almost everywhere differentiable we conclude from \eqref{eq:mu-first-integral} that
\begin{align*}
		\dot\mu_t ~\le~ - \frac{\dot\eta_t - \delta \kappa_t}{\eta_t}\mu_t, \qquad \text{ for a.e. } t\in [t_0, T].
\end{align*}
Hence,  it follows from the modified Gr\"onwall's inequality in Lemma~\ref{lem:lower_gronwall}
\begin{align*}
		\mu_t ~\ge~\mu_{T}\exp\left(\int_{t}^{T}\frac{\dot\eta_s - \delta \kappa_s}{\eta_s} \, ds\right)  
		~\ge~  \mu_T\frac{\eta_T }{\eta_t}e^{-K_1(T-t)} \quad t \in [t_0, T].
\end{align*} 
In particular, $\mu_{t_0} > 0$. From the definition of $t_0$ and the continuity of $\mu$ it thus follows that $t_0 = 0$. 
An analogous argument applies to the case $\mu_T < 0$.
\end{proof}

The fact that a fixed-point $\mu$ does not change sign implies that either $\overline{f}_\mu \equiv f_\mu$ and $\underline{f}_\mu \equiv 0$ or $\overline{f}_\mu \equiv 0$ and $\underline{f}_\mu \equiv f_\mu$, which considerably simplifies the equation \eqref{eq:mu-first-integral}. The following proposition shows that the sign of $\mu$ depends only on the sign of the average initial position $\E[\nu_0]$.

\begin{proposition}\label{lem:nu0-sign}
Let $\mu\in L^{1}([0,T]; \bR)$ be a solution to the fixed-point equation \eqref{eq:consistency}. Then, $$\sign(\mu_t) = \sign(\E[\nu_0]), \quad 0 \le t \le T.$$
In case $\E[\nu_0] > 0$, then $\mu$ is a solution to the fixed-point equation \eqref{eq:consistency} if and only if it satisfies the equation\footnote{Due to symmetry, if $\E[\nu_0] < 0$ we equivalently can solve for the fixed-point problem with reflected initial measure $\widetilde{\nu_0}((-\infty, x]) = \nu_0((-x, \infty])$ and $\E[\widetilde{\nu_0}] > 0$.}
\begin{equation} \label{eq:mu-second-integral}
\begin{split}
	\mu_t  =&~ \mu_T +\int_{t}^{T}\frac{\kappa_s}{\eta_s}\left(q_0\big( f_\mu(s)\big) + p_0(0) \right)\mu_s\, ds
		+\int_t^{T}\frac{\lambda_s + \delta\dot\kappa_s}{\eta_s}\left(\int_s^{T}\mu_u du\right)ds \\
		&~+ \int_t^{T}\frac{\dot\eta_s - \delta \kappa_s}{\eta_s}\mu_s\;ds, \qquad t\in[0,T],
\end{split}
\end{equation}
subject to the the terminal condition
\begin{align}\label{eq:mu_terminal_value}
\mu_T = \frac{\alpha_T^\delta}{\eta_T} \left(\E[\nu_0] -(1-q_0(0)) x_\mu  - \int_0^{x_\mu} q_0(x)dx\right),
\end{align}
where $x_\mu := f_\mu(T)$.
\end{proposition}
\begin{proof} By Proposition~\ref{prop:fix-point-integral} $\mu$ is a fixed-point of \eqref{eq:consistency} if and only if it solves the integral equation \eqref{eq:mu-first-integral} and satisfies $F(\mu)_T = \mu_T$.
According to Lemma~\ref{lem:mu-sign} the sign of $\mu$ is then determined by the sign of the terminal value $\mu_T$.
In order to analyze the terminal value $\mu_T$ we will take the limit for $t\nearrow T$ in the equation \eqref{eq:consistency}.
To this end, recall the limiting behavior of $Y^{\delta,x,\mu}_t$ from \eqref{limY} and note that
\begin{align*}
\lim_{t \nearrow T}\xi^{\delta,x, \mu}_t = \lim_{t \nearrow T}\frac{Y^{\delta, x, \mu}_t-\delta\kappa_t X^{\delta, x, \mu}_t}{\eta_t} = \frac{\alpha^{\delta}_T}{\eta_T}( x- x_\mu ) \mathbf{1}_{\{x\in I_\mu(T)\}},\qquad x\in \bR.
\end{align*}
Further, recalling the a priori estimate on $\xi^{\delta, x, \mu}$ from Lemma~\ref{lem:a_priori_candiate} we obtain from dominated convergence
\begin{equation*}
	\begin{split}
		\mu_T= F(\mu)_T =&~\lim_{t\nearrow T}\int_{\bR}\xi^{\delta, x, \mu}_t\nu_0(dx) \\
		=&~\int_{I_\mu(T)}\frac{\alpha^{\delta}_T( x- x_\mu )}{\eta_T}\nu_0(dx)\\
		=&~\frac{\alpha^{\delta}_T}{\eta_T} \left(\E[\nu_0] - x_\mu - \int_{\underline{f}_\mu(T)}^{\overline{f}_\mu(T)} (x-x_\mu)\nu_0(dx) \right).
	\end{split}
\end{equation*}
Now assume that $\mu_T > 0$.
Then it follows from Lemma \ref{lem:mu-sign} that $\mu_t > 0$ for all $t \in [0,T]$.
Hence $\overline{f}_\mu \equiv f_\mu$ and $\underline{f}_\mu \equiv 0$; in particular $x_\mu > 0$.
Therefore, the above expression for the terminal value simplifies to
\begin{align*}
	\mu_T &=~ \frac{\alpha^{\delta}_T}{\eta_T} \left(\E[\nu_0] - x_\mu - \int_{0}^{x_\mu} (x-x_\mu)\nu_0(dx) \right) \\
	&=~ \frac{\alpha^{\delta}_T}{\eta_T} \left(\E[\nu_0] -(1-q_0(0)) x_\mu  - \int_0^{x_\mu} q_0(x)dx\right).
\end{align*}
From this we conclude that $\E[\nu_0] > 0$ since otherwise all the terms in the bracket would be non-positive. Hence, $\mu$ satisfies the equation \eqref{eq:mu-second-integral}.
Analogously, we see that $\mu_T < 0$ implies $\bE[\nu_0] < 0$. Finally, if $\mu_T = 0$, then by Lemma~\ref{lem:mu-sign} it holds $\mu \equiv 0$ and therefore $x_\mu = \underline{f}_\mu(T) = \overline{f}_\mu(T) = 0$, hence $\E[\nu_0] = 0$.
\end{proof}

\subsection{Existence and uniqueness of solutions}
In this section we  prove that the integral equation \eqref{eq:mu-second-integral} with terminal condition \eqref{eq:mu_terminal_value} has a unique solution; by Proposition~\ref{lem:nu0-sign} this implies that the fixed-point problem \eqref{eq:consistency} has a unique solution. 

Apart from the non-linearity induced by the function $q_0$ the main difficulty in solving the equation is that it involves forward and backward integrals of $\mu$.
More precisely,  
\begin{align*}
f_\mu(t) = \int_0^{t}\kappa_s h^{\delta}_s \mu_s \, ds, \qquad t \in[0,T]
\end{align*}
is defined as a forward integral, while all other integrals in \eqref{eq:mu-second-integral} are backward integrals. Furthermore, the terminal value of $\mu$ is endogenous as it depends on $x_\mu = f_\mu(T)$. To overcome these problems we introduce the following backward integral equations 
\begin{equation} \label{eq:mu-third-integral}
\begin{split}
	\mu_t  =&~ \mu^{c}_T +\int_{t}^{T}\frac{\kappa_s}{\eta_s}\left(q_0\!\left(c - \int_s^{T}\kappa_r h^{\delta}_r \mu_r \, dr \right) + p_0(0) \right)\mu_s\, ds
		+\int_t^{T}\frac{\lambda_s + \delta\dot\kappa_s}{\eta_s}\left(\int_s^{T}\mu_u du\right)ds \\
		&~+ \int_t^{T}\frac{\dot\eta_s - \delta \kappa_s}{\eta_s}\mu_s\;ds, \qquad t\in[0,T],
\end{split}
\end{equation}
for any $c \in \mathbb{R}$ with exogenous terminal condition 
\begin{align*}
\mu^{c}_T := \frac{\alpha^\delta_T}{\eta_T} \left(\E[\nu_0] -(1-q_0(0)) c  - \int_0^{c} q_0(x)dx\right).
\end{align*}
A solution $\mu^{c}$ to the equation \eqref{eq:mu-third-integral} yields a solution to the original equation \eqref{eq:mu-second-integral} if the additional condition $f_{\mu^c}(T) = c$ is satisfied.
The following result shows that such a $c$ always exists and is unique.

\begin{theorem}\label{thm:mfg-ode-solution}
\begin{itemize}
	\item[i)] For each $c \in \mathbb{R}$ there exists a unique solution $\mu^{c}$ in $C^{0}([0,T])$ to the integral equation \eqref{eq:mu-third-integral}.
	\item[ii)] There exists a unique $\widehat{x} \in \mathbb{R}$ such that $f_{\mu^{\hat{x}}}(T) = \hat{x}$. Moreover, $\mu^{\hat x}$ is the unique solution to \eqref{eq:mu-second-integral} and it holds that $0 < \widehat{x} < \sup(\mathrm{supp}(\nu_0))$.\footnote{The quantity $\hat x$ can be understood as the largest initial position for which early liquidation is beneficial.}
\end{itemize}
\end{theorem}
\begin{proof}
\begin{itemize}
	\item[i)] By Assumption~\ref{ass:single-player} the coefficients $\eta, \kappa, \lambda$ and $1/\eta$ are non-negative and bounded, and the derivatives $\dot\kappa$ and $\dot\eta$ are bounded.
	By Lemma~\ref{lem:about-h} the functions $h^{\delta}$ and $\dot h^{\delta}$ are bounded uniformly in $\delta\in[0,1]$ and $1/h^{\delta}$ is bounded on compact subsets of $(0,T]$ uniformly for $\delta\in[0,1]$.
	
	In what follows $K$ denotes a positive constant that may change from line to line, but only depends on the aforementioned bounds.
	 
	We will first consider the linear and non-linear parts of equation \eqref{eq:mu-third-integral} separately.
	Similarly to the proof of Proposition~\ref{prop:fix-point-integral}, we define the linear operator $G$ mapping $C^{0}([0,T])$ into the space of absolutely continuous functions on $[0,T]$ by
	\begin{align*}
	G(\mu)_t  =&~ \int_t^{T}\frac{\lambda_s + \delta\dot\kappa_s}{\eta_s}\left(\int_s^{T}\mu_u du\right)ds + \int_t^{T}\frac{\dot\eta_s + (p_0(0) - \delta )\kappa_s }{\eta_s}\mu_s\;ds, \qquad t\in[0,T].
\end{align*}
We readily see that the following estimate holds for all $\mu\in C^{0}([0,T])$:
\begin{align}\label{eq:G_estimate}
\vert G(\mu)_t \vert \le K\int_t^{T}\vert \mu_s \vert \, ds, \qquad t \in [0,T].
\end{align}
The remaining parts of \eqref{eq:mu-third-integral} are described by the non-linear operator $H_c$ mapping $C^{0}([0,T])$ into the space of absolutely continuous functions on $[0,T]$ defined by
\begin{align*}
	H_c(\mu)_t  =&~ \mu^{c}_T +\int_{t}^{T}\frac{\kappa_s}{\eta_s}q_0\!\left(c - \int_s^{T}\kappa_u h^{\delta}_u \mu_u \, du \right) \mu_s\, ds, \qquad t\in[0,T], \quad c\ge0.
\end{align*} 
Since $0 \le q_0 \le 1$ we immediately obtain the following estimate:
\begin{align} \label{eq:H_estimate}
	\vert H_c(\mu)_t \vert  \leq &~ \vert \mu^{c}_T| + K \int_{t}^{T} \vert\mu_s\vert\, ds, \qquad t\in[0,T], \quad c\ge0.
\end{align}

Since $H_c$ is non-linear, this estimate is not sufficient to prove the existence result.
Introducing the integrated tail probability $Q_0(x) := \int_0^{x}q_0(y)\,dy$, integration by parts yields the following representation for all $\mu\in C^{0}([0,T])$
\begin{align*}
H_c(\mu)_t =~&\mu^c_T+ Q_0(c)\frac{1}{h^\delta_T\eta_T}-Q_0\left(c-\int_t^{T}\kappa_s h^\delta_s \mu_s ds\right)\frac{1}{h^\delta_t \eta_t} \\
&+ \int_t^{T} Q_0\left(c-\int_s^{T}\kappa_u h^{\delta}_u \mu_u du \right) \frac{\dot h^{\delta}_s \eta_s + h^{\delta}_s \dot\eta_s}{(h^{\delta}_s)^{2} \eta_s^{2}} ds, \qquad t\in (0, T].
\end{align*}
Note that $Q_0$ is globally Lipschitz continuous with coefficient $\Vert Q_0 \Vert_{\mathrm{Lip}} = \Vert q_0 \Vert_\infty = 1$ and similarly also the mapping $c \mapsto \mu^{c}_T$ is globally Lipschitz continuous.

Since $1/h^{\delta}$ is only bounded on compact sets of $(0,T]$ we introduce the time $t_k := 1/k$ for some $k\in\mathbb{N}$.
For any $\mu, \widetilde{\mu} \in C^{0}([t_k,T])$ and $c, \widetilde{c} \ge 0$ we obtain the following estimate:
\begin{equation}\label{eq:H_estimate_fine}
\begin{split}
\vert H_c(\mu)_t - H_{\widetilde c}(\widetilde\mu)_t\vert \le&~ \vert \mu^{c}_T - \mu^{\widetilde c}_T\vert + \frac{1}{h^\delta_T\eta_T}\vert c - \widetilde c\vert + \frac{1}{h^{\delta}_t \eta_t}\left(\int_t^{T}\kappa_s h^{\delta}_s \vert \mu_s - \widetilde\mu_s\vert \,ds  + |c-\widetilde c| \right)\\
&~+ \int_t^{T} \left\vert\frac{\dot h^{\delta}_s \eta_s + h^{\delta}_s \dot\eta_s}{(h^{\delta}_s)^{2} \eta_s^{2}} \right\vert\left( \int_s^{T}\kappa_u h^{\delta}_u \vert \mu_u -  \widetilde\mu_u \vert \,du+  |c-\widetilde c| \right)ds \\
\le&~  K \vert c - \widetilde{c} \vert + K\int_t^{T} \vert \mu_s - \widetilde{\mu}_s\vert ds, \qquad t\in[t_k,T].
\end{split}
\end{equation}
Now let $c\ge 0$ be arbitrary but fixed. 
Iterating the estimates \eqref{eq:G_estimate} and \eqref{eq:H_estimate_fine}, we see that for $n\in\bN$  it holds that
\begin{align*}
\left\vert [H_c + G]^{n}(\mu)_t - [H_c + G]^{n}(\widetilde\mu)_t\right\vert 
&\le~ K\int_t^{T} \left\vert [H_c + G]^{n-1}(\mu)_t - [H_c + G]^{n-1}(\widetilde\mu)_t\right\vert \, ds \\
&\le~ \frac{K^{n}T^{n}}{n!}\Vert\mu - \widetilde{\mu}\Vert_\infty, \qquad t\in[t_k,T].
\end{align*}
From \cite[Theorem 2.4]{Teschl-2016} it then follows that $[H_c + G]$ has a unique fixed-point $\mu^{c,k}$ in $C^{0}({[t_k,T]})$. 

From the uniqueness it follows that the pointwise limit $\mu^{c}_t := \lim_{k \to \infty}\mu_{t}^{c,k}$ is well defined and satisfies $[H_c + G](\mu^{c})_t = \mu^{c}_t$ for all $t\in(0,T]$.
From the estimates \eqref{eq:G_estimate}, \eqref{eq:H_estimate} and  Gr\"onwall's inequality it then follows that
\begin{align}\label{eq:existence-mu-upper}
\vert\mu^{c}_t\vert  ~\le~ |\mu^{c}_T| + K\int_t^{T}\vert \mu^{c}_s \vert \,ds ~ \le ~|\mu^{c}_T|e^{K(T-t)}\qquad t\in(0,T].
\end{align}
Hence $\mu^{c}$ is bounded which allows us to extend $\mu^{c}$ to a continuous function on $[0,T]$ using dominated convergence in the representation $\mu^{c}_t =[H_c + G](\mu^{c})_t$.
By construction it follows that $\mu^{c}$ is the unique fixed-point of $[H_c + G]$ in $C^{0}([0,T])$, hence the unique solution to the equation \eqref{eq:mu-third-integral}.

\item[ii)] We are first going to prove the existence of $\hat{x}$ and afterwards we will prove that $\mu^{\hat{x}}$ is the unique solution to equation \eqref{eq:mu-second-integral}.
  For all $c\ge0$ we denote by $\mu^{c}$ the unique solution to \eqref{eq:mu-third-integral}. Our aim is to show that the map
\begin{align*}
\psi: [0, \infty) \to \bR,  \quad c \mapsto c - f_{\mu^{c}}(T) = c - \int_0^{T}h^{\delta}_s \kappa_s \mu^{c}_s \, ds,
\end{align*}
has a unique root.
Note that from the above estimates \eqref{eq:G_estimate} and \eqref{eq:H_estimate_fine} we obtain that for any $t_0 \in (0, T]$ it holds
\begin{align*}
\vert \mu^{c_1}_t - \mu^{c_2}_t\vert &=~ \Vert [H_{c_1} + G](\mu^{c_1}) - [H_{c_2} + G](\mu^{c_2})\Vert_\infty \\
&\le~ K\vert c_1 - c_2 \vert + K\int_t^{T}\vert  \mu^{c_1}_s - \mu^{c_2}_s\vert\, ds \\
&\le~ K\vert c_1 - c_2 \vert e^{K(T-t)}, \qquad t \in [t_0, T],
\end{align*}
where the last estimate follows from Gr\"onwall's inequality.
Since additionally the estimate \eqref{eq:existence-mu-upper} holds we can use dominated convergence to conclude that the map $\psi$ is continuous.

To prove the existence of a root it suffices to prove that $\psi$ changes sign.
To this end we note that the bounds of \eqref{eq:mu-sign} also hold for solutions to the equation \eqref{eq:mu-third-integral}; in particular, $\sign(\mu^{c}) = \sign(\mu^{c}_T)$ for all $c\ge0$. 
Indeed, an inspection of the proof of Lemma~\ref{eq:mu-sign} shows that only the estimate $0 \le q_0 \le 1$ is used to estimate the factor involving the $q_0$ term.
We now need to distinguish two cases. 

\begin{itemize}
	\item $q_0(0) < 1$. By the continuity and the monotonicity of the mapping $c \mapsto \mu^c_T$ there exists a unique $c_0 > 0$ such that $\mu^{c_0}_T = 0$ and  we see that the unique solution to the equation \eqref{eq:mu-third-integral} is trivial in this case, i.e., $\mu^{c_0} \equiv 0$.
	Therefore it holds $\psi(c_0) = c_0 > 0$.
	
	On the other hand, since $\mu^0_T > 0$ it follows that $\mu^{0} > 0$ and as a result,
	\[
		\psi(0) = -\int_0^T h^{\delta}_s \kappa_s \mu^0_s ds < 0. 
	\]  
	\item $q_0(0) = 1$.  In this case, there are only sellers in the market and
	\[
		\mu^c_T = \frac{\alpha^\delta_T}{\eta_T}\left( \mathbb E[\nu_0] - \int_0^c q_0(x) \,dx\right)  \to 0 \quad \mbox{ for } \quad c \to \infty;\footnote{If $\tilde c := \sup(\mathrm{supp}(\nu_0)) < \infty$, then this equation is to be understood as $\mu^{\tilde c}_T=0$. } 
	\]
	In particular, it follows again from \eqref{eq:mu-sign} that
	\[
		\Vert \mu^c\Vert_\infty  \to 0 \quad \mbox{ for } \quad c \to \infty. 
	\]
	Thus, by dominated convergence there exists $c_0 > 0$ such that
	\[
		\psi(c_0) = c_0 -\int_0^T h^{\delta}_s \kappa_s \mu^{c_0}_s ds > 0.
	\]
	At the same time we obtain analogously to the first case that $\psi(0)<0$.
\end{itemize}
Hence, in both cases $\psi$ changes its sign and the existence of at least one root $\hat{x}$ follows.
To prove the uniqueness of $\hat{x}$ we show that the mapping $\psi$ is strictly increasing. To this end, we fix $c_1 > c_2$ and define
\begin{align*}
f^{c_i}_t := c_i - \int_t^{T} h^{\delta}_s \kappa_s \mu^{c_i}_s \, ds, \qquad t\in[0,T], \quad i = 1,2.
\end{align*}
Note that by definition it holds $\mu^{c_1}_T < \mu^{c_2}_T$ and $f^{c_1}_T > f^{c_2}_T$.   Now let us denote by $t^\mu_0$ and $t_0^f$ the last times when $\mu^{c_2}-\mu^{c_1}$, respectively $f^{c_1}-f^{c_2}$ changes its sign, i.e.,
\begin{align*}
	t_0^\mu &= \inf\big\{t\in[0,T] \;\big\vert\; \mu^{c_2}_s-\mu^{c_1}_s>0, \; s\in[t,T]\big\},\\
	 t_0^f &= \inf\big\{t\in[0,T] \;\big\vert\; f^{c_1}_s-f^{c_2}_s>0, \; s\in[t,T]\big\}. 
\end{align*}
For any point $t\in[t_0^\mu\vee t^f_0,T]$ where $\mu^{c_1}$ and $\mu^{c_2}$ are differentiable it holds that
\begin{equation*}
	\begin{split}
			\dot\mu^{c_2}_t-\dot\mu^{c_1}_t\leq &~ \frac{\kappa_t}{\eta_t}\left\{q_0(f^{c_1}_t)\mu^{c_1}_t-q_0(f^{c_2}_t)\mu^{c_1}_t\right\}+\frac{\kappa_t}{\eta_t}\left\{q_0(f^{c_2}_t)\mu^{c_1}_t-q_0(f^{c_2}_t)\mu^{c_2}_t\right\}    -\frac{\dot\eta_t  - \delta \kappa_t}{\eta_t} (\mu^{c_2}_t-\mu^{c_1}_t)\\
			\leq&~-\frac{\dot\eta_t - \delta \kappa_t}{\eta_t} (\mu^{c_2}_t-\mu^{c_1}_t).
	\end{split}
\end{equation*} 
Since $\mu^{c_1}$ and $\mu^{c_2}$ are almost everywhere differentiable it follows by the version of Gr\"onwall's inequality in Lemma~\ref{lem:lower_gronwall} that
\begin{equation}\label{estimate:diff-mu}
	\mu^{c_2}_t-\mu^{c_1}_t\geq  (\mu^{c_2}_T-\mu^{c_1}_T)\exp\left( \int_t^{T}\frac{\dot\eta_s - \delta \kappa_s}{\eta_s}\, ds\right)>0, \qquad t\in[t_0^\mu\vee t^f_0,T].
\end{equation}
From this it also follows
\begin{equation}\label{estimate:diff-f}
			\dot f^{c_1}_t-\dot f^{c_2}_t=h_t\kappa_t(\mu^{c_1}_t-\mu^{c_2}_t)<0, \qquad t\in[t_0^\mu\vee t^f_0,T].
\end{equation}

If $t_0^\mu\leq t^f_0$,  then the inequality \eqref{estimate:diff-f} implies that $t\mapsto f^{c_1}_t-f^{c_2}_t$ is strictly decreasing on the interval $[t^f_0,T]$. 
By the definition of $t_0^f$ it then holds that $t_0^f=0$, which implies that $t_0^\mu=0$ and thus  $t\mapsto f^{c_1}_t-f^{c_2}_t$ is strictly decreasing on $[0,T]$.
In particular it holds that  $$\psi(c_1) - \psi(c_2) = f_0^{c_1}-f^{c_2}_0>f_T^{c_1}-f^{c_2}_T >0$$ and so the map $\psi$ is strictly increasing. 

If $t_0^\mu \geq t^f_0$, then the inequality \eqref{estimate:diff-mu} implies that $\mu^{c_2}_t-\mu^{c_1}_t>0$ on the interval $[t_0^\mu,T]$. 
Thus, by the definition of $t_0^\mu$ it must hold  $t_0^\mu=0$, which also implies that $t^f_0=0$. 
Using the inequality \eqref{estimate:diff-f} again, we conclude that $c\mapsto f^c_0$ is also strictly increasing in this case.

We are now proving the second part of the assertion in ii).
From the definition of $\hat{x}$ it follows that 
\begin{align*}
\hat{x} - \int_t^{T} h^{\delta}_s \kappa_s \mu^{\hat x}_s \, ds =  \int_0^{t} h^{\delta}_s \kappa_s \mu^{\hat x}_s \, ds = f_\mu^{\hat x}(t), \qquad t\in[0,T].
\end{align*}
Inserting this relation into \eqref{eq:mu-third-integral} for $\mu^{\hat x}$ we immediately see that $\mu^{\hat x}$ also solves  \eqref{eq:mu-second-integral}.
Conversely, if $\mu$ is a solution to \eqref{eq:mu-second-integral}, then defining $c:=f_\mu(T)$ we immediately see that $\mu$ also solves \eqref{eq:mu-third-integral} and by uniqueness it must holds $\mu = \mu^{c}$. 
Since then $f_{\mu^{c}}(T) = f_\mu(T) = c$ it follows from uniqueness of $\hat{x}$ that $\mu = \mu^{\hat x}$.
\end{itemize}
\end{proof}

\section{Equilibrium analysis}

For MFGs the verification argument given in Section~\ref{sec:single_player_verification} along with the solvability of the fixed-point problem readily yields the existence of a unique equilibrium for our liquidation game. 

\begin{theorem}\label{thm:mfg-final}
Under Assumption~\ref{ass:single-player} and Assumption \ref{ass:mean-field}, the mean-field liquidation game admits an equilibrium. 
{In case $\E[\nu_0]\neq 0$, the equilibrium is unique in the class of equilibria with continuous aggregate trading rates $\mu\in C^{0}([0,T])$ that satisfy $\mu_T \neq 0$. \footnote{{In case $\E[\nu_0] = 0$ it follows from the existence part of the proof that $\mu\equiv0$ yields an equilibrium, which is the most natural equilibrium in this case.
In general, we cannot rule out the existence of other equilibria that might change their sign infinitely often when approaching the terminal time, but a simple extension of the proof would give uniqueness among those equilibria that change sign only finitely many times.}}}
\end{theorem}
\begin{proof}
According to Theorem~\ref{thm:mfg-ode-solution} there exists a unique solution $\mu\in C^{0}([0,T])$ to the equation \eqref{eq:mu-second-integral}. By Proposition~\ref{prop:fix-point-integral} it follows that $\mu$ is the unique solution to the fixed-point problem \eqref{eq:consistency} and that $\sign(\mu) \equiv \sign(\E[\nu_0])$. Therefore, it follows from Proposition~\ref{prop:verification} that $\xi^{\delta, x, \mu}$ is the unique optimal liquidation strategy for the representative player with initial position $x \in \bR$. Since $\mu$ solves the fixed-point problem \eqref{eq:consistency}, we conclude that $(\xi^{\delta, x, \mu})_{x\in\bR}$ is an equilibrium of the MFG.

Now {assume that $\E[\nu_0] \neq 0$ and} let $(\widetilde{\xi}^{x})_{x\in \bR}$ be another equilibrium of the MFG, such that the aggregated trading rate $\widetilde{\mu}$ is continuous {and satisfies $\widetilde{\mu}_T \neq 0$}. 
Then there exists a minimal $t_0\in[0,T]$ such that $\widetilde{\mu}\vert_{[t_0, T]} \in L^1_-([0,T])\cup L^1_+([0,T])$.
{Note that} $\widetilde\mu|_{[t_0,T]}$ is a solution to the MFG restricted on $[t_0, T]$ with initial measure $\nu_{t_0} = Law((\widetilde{X}^{x}_{t_0})_{x\in\bR})$. 
By the choice of $t_0$, it thus follows from Proposition~\ref{prop:verification}  that $\widetilde{\xi}^{x}\vert_{[t_0, T]} = \xi^{\delta, \widetilde{\mu}, x}\vert_{[t_0, T]}$.  As a result, $\widetilde{\mu}\vert_{[t_0, T]}$ solves the fixed-point problem \eqref{eq:consistency} with initial measure $\nu_{t_0}$ and we eventually see that $\sign(\widetilde{\mu}_t) = \sign(\E[\nu_{t_0}])$ for all $t\in[t_0, T]$. From the minimality of $t_0$ it follows that $t_0 = 0$, and we conclude from the uniqueness of solutions to \eqref{eq:consistency} that $\widetilde{\mu} = \mu$.
\end{proof}

\subsection{Equilibria in the $N$-player game}\label{sec:N-player}

Establishing the existence of equilibria in the $N$-player liquidation game is more complex, due to the lack of a general verification result. To obtain a convex optimization problem in the $N$-player game we need to lightly strengthen our assumptions on the cost parameters.
\begin{ass}\label{ass:npg}
In addition to Assumption~\ref{ass:single-player}, we assume that the cost parameters satisfy $\eta - \delta\kappa >0$ and $\lambda - \delta\kappa \geq 0$, with $\delta=\frac{1}{N}$.
\end{ass}

In the $N$-player game, we denote the initial distribution of the players'  positions by $\nu_0=\frac{1}{N}\sum_{j=1}^N\delta_{x_j}$, which trivially satisfies Assumption \ref{ass:mean-field}. The following theorem establishes the existence of an equilibrium in the $N$-player game. While some estimates in the verification argument seem similar to the MFG, the situation is different. In the MFG we verified that the candidate strategy is a best response to any given mean trading rate $\mu$ that does not change sign. In our current setting this argument only applies in equilibrium, which complicates the uniqueness argument. 

In what follows we put
\[
	\xi^{\mu,i}:= \xi^{\frac 1 N, x_i, \mu}, \quad \xi^{\mu,-i}:= \big( \xi^{\frac 1 N, x_j, \mu} \big)_{j \neq i}, \quad 
	X^{\mu,i} = X^{\frac 1 N, x_i,\mu}, \quad Y^{\mu,i} = Y^{\frac 1 N, x_i,\mu}
\]

\begin{theorem}\label{thm:N-player-verification}
Let $\mu\in L^{1}([0,T])$ be the unique solution to \eqref{eq:consistency}.
Then, for every $i = 1, \dots, N$ the admissible control $\xi^{\mu,i}$ is the unique solution to the optimal control problem
\begin{align}\label{eq:problem-player-very}
\inf_{\xi \in \mathcal{A}_{x_i}}J(\xi;\xi^{\mu, -i}),  \qquad
X_t = x_i - \int_0^{t}\xi_s \, ds, \qquad t\in[0,T].
\end{align}
As a result, $(\xi^{\mu, 1}, \dots, \xi^{\mu, N})$ is a Nash equilibrium of the $N$-player game.
\end{theorem}
\begin{proof}
{By Lemma~\ref{lem:nu0-sign} it follows that $\mu \in L^{1}_-([0,T])\cap L^{1}_+([0,T])$ and therefore by Lemma~\ref{lem:a_priori_candiate} it follows that $\xi^{\mu, i} \in \mathcal{A}_{x_i}$ for all $i=1, \dots, N$.}
Now for each $i=1, \dots, N$ let $\xi \in \mathcal{A}_{x_i}$ be an admissible control with corresponding state process $X$ and absorption time $\tau$. Let $\tau^i:=\tau_\mu(x^i)$. We distinguish again two cases. 

\begin{enumerate}[label=(\roman*)]
\item $\tau \le \tau^{i}$. %
Using the fact that $\mu$ satisfies the equation \eqref{eq:consistency} shows that
\begin{align*}
		J(\xi; \xi^{\mu,-i})&-J(\xi^{\mu, i}; \xi^{\mu,-i}) \\
		=&~\int_0^{\tau}\left(	\frac{1}{2}\eta_t\xi^2_t+\kappa_t\left( \frac{1}{N}\xi_t + \frac{1}{N}\sum_{j\neq i}\xi^{\mu, j}_t \right)X_t+\frac{1}{2}\lambda_t X^2_t			\right)\,dt \\
		&\quad- \int_0^{\tau^{i}}\left(	\frac{1}{2}\eta_t(\xi^{\mu, i}_t)^2+\kappa_t\left( \frac{1}{N}\xi^{\mu, i}_t + \frac{1}{N}\sum_{j\neq i}\xi^{\mu, j}_t \right)X^{\mu, i}_t+\frac{1}{2}\lambda_t (X^{\mu, i}_t)^2			\right)\,dt\\
		= &~\int_0^{\tau^{i}}\left(	\frac{1}{2}\left(\eta_t-\frac{\kappa_t}{N}\right)\xi^2_t+\frac{1}{2}\frac{\kappa_t}{N} \left(\xi_t + X_t\right)^{2}+\frac{1}{2}\left(\lambda_t-\frac{\kappa_t}{N}\right) X^2_t \right)\,dt \\
		&\quad- \int_0^{\tau^{i}}\left(	\frac{1}{2}\left(\eta_t-\frac{\kappa_t}{N}\right)(\xi^{\mu, i}_t)^2+\frac{1}{2}\frac{\kappa_t}{N} \left(\xi^{\mu, i}_t + X^{\mu, i}_t\right)^{2}+\frac{1}{2}\left(\lambda_t-\frac{\kappa_t}{N}\right)(X^{\mu, i}_t)^2	\right)\,dt\\	
		&\quad + \int_0^{\tau^{i}}\kappa_t\left(\frac{1}{N}\sum_{j\neq i}\xi^{\mu, j}_t\right)(X_t - X^{\mu,i}_t) dt\\
			\ge&~ \int_0^{\tau^{i}}\left(	\left(\eta_t-\frac{\kappa_t}{N}\right)\xi^{\mu, i}_t(\xi_t - \xi^{\mu, i}_t)+\frac{\kappa_t}{N} (\xi^{\mu, i}_t + X^{\mu, i}_t)\left(\xi_t - \xi^{\mu, i}_t + X_t -  X^{\mu, i}_t \right)\right)\,dt\\
		&~+ \int_0^{\tau^{i}}\left(\lambda_t-\frac{\kappa_t}{N}\right)X^{\mu, i}_t(X_t - X^{\mu, i}_t)\,dt
		+\int_0^{\tau^{i}}\kappa_t\left(\frac{1}{N}\sum_{j\neq i}\xi^{\mu, j}_t\right)(X_t - X^{\mu,i}_t) dt\\
		\geq &~ \int_0^{\tau^{i}}\left(	\left(\eta_t\xi^{\mu, i}_t + \frac{\kappa_t}{N}X^{\mu, i}_t\right)(\xi_t - \xi^{\mu, i}_t) +\left(\lambda_tX^{\mu, i}_t +\kappa_t\mu_t \right)(X_t - X^{\mu, i}_t) \right)\,dt.
\end{align*}
Since $\Vert Y^{\mu, i}\Vert_\infty < \infty$ and $X^{\mu, i}_{\tau^{i}} = X_{\tau^{i}} = 0$,
the same integration by parts argument given in the proof of Proposition \ref{prop:verification} yields that
\begin{align}\label{eq:N-very-integration-by-parts}
		0=Y^{\mu, i}_{\tau^{i}}(X_{\tau^{i}} -X^{\mu, i}_{\tau^{i}}	)
		=-\int_0^{\tau^{i}}Y^{\mu, i}_t(\xi_t-\xi^{\mu, i}_t)\,dt-\int_0^{\tau^{i}}(\kappa_t\mu_t+\lambda_t X^{\mu, i}_t)(X_t-X^{\mu, i}_t)\,dt.
\end{align}
Using that $(X^{\mu, i}, Y^{\mu,i})$ solves the equation \eqref{eq:FBODE} with $\delta=\frac{1}{N}$ this shows that %
\begin{equation*}
		J(\xi; \xi^{\mu,-i})-J(\xi^{\mu, i}; \xi^{\mu,-i}) \ge \int_0^{\tau^{i}}	\left(\eta_t\xi^{\mu, i}_t + \frac{\kappa_t}{N}X^{\mu, i}_t - Y^{\mu, i}_t\right)(\xi_t - \xi^{\mu, i}_t)\,dt = 0.
\end{equation*}

\item $\tau > \tau^{i}$. In this case, $\tau^{i} < T$ and $\xi^{\mu, i}_t = Y^{\mu,i}_t = 0$ for all $t\in [\tau, T]$. Since $\tau^i = T$ if $x^i$ and $\mu_T$ have different signs, $\tau > \tau^{i}$ holds only if $x^i$ and $ \mu_T $ have the same sign in which case 
\[
	\mu_t X_t \geq 0 \quad \mbox{ for all } t \in [0,T].
\]  
Moreover, the integration by parts argument still applies since $Y^{\mu, i}_{\tau^{i}} = 0$. Hence, 
\begin{align*}
J(\xi; \xi^{\mu,-i})-J(\xi^{\mu, i}; \xi^{\mu,-i})\ge&~\int_{\tau^{i}}^{\tau}\left(	\frac{1}{2}\left(\eta_t-\frac{\kappa_t}{N}\right)\xi^2_t+\frac{1}{2}\frac{\kappa_t}{N} \left(\xi_t + X_t\right)^{2}+\frac{1}{2}\left(\lambda_t-\frac{\kappa_t}{N}\right) X^2_t \right)\,dt \\		
		&\quad + \int_{\tau^{i}}^{\tau}\kappa_t\left(\frac{1}{N}\sum_{j\neq i}\xi^{\mu, j}_t\right)X_t  dt\\
		\ge&~ \int_{\tau^{i}}^{\tau}\kappa_t\left(\frac{1}{N}\sum_{j\neq i}\xi^{\mu, j}_t\right)X_t  dt. 
\end{align*}

The key observation is now that $\xi^{\mu, i}_t = 0$ for $t \in [\tau^{i}, T]$ yields 
\begin{align*}
\frac{1}{N}\sum_{j\neq i}\xi^{\mu, j}_t = \frac{1}{N}\sum_{j=1}^{N}\xi^{\mu, j}_t = \mu_t, \qquad t \in [\tau^{i}, T].
\end{align*}
Hence, 
\begin{align*}
J(\xi; \xi^{\mu,-i})-J(\xi^{\mu, i}; \xi^{\mu,-i})\ge&~ \int_{\tau^i}^{\tau}\mu_tX_t  dt \ge 0. 
\end{align*}
\end{enumerate}
\end{proof} 

We are now going to show that there exists a unique Nash equilibrium with continuous aggregate trading rate. To do so, we first establish a uniqueness of equilibrium result when the aggregate trading rate does not change sign.
 
\begin{theorem}\label{thm:N-sign-uniqueness}
	\begin{enumerate}
		\item
Let $(\xi^{\ast, 1}, \dots, \xi^{\ast, N}) \in \mathcal{A}_{x_1} \times \dots, \mathcal{A}_{x_N}$ be a Nash equilibrium of the $N$-player game, such that the average trading rate  does not change sign:
$$\mu := \frac{1}{N}\sum_{i=1}^{N} \xi^{\ast, i} \in L^1_-([0,T])\cup L^1_+([0,T]).$$
Then $\xi^{\ast, i} = \xi^{\mu, i}$ for all $i=1, \dots, N$ and  $\mu$ is the unique solution to the fixed-point equation \eqref{eq:consistency}.
\item{In case $\E[\nu_0]\neq 0$, the equilibrium is unique in the class of equilibria with continuous aggregate trading rates $\mu\in C^{0}([0,T])$ that satisfy $\mu_T \neq 0$}.
\end{enumerate}
\end{theorem} 
\begin{proof}
\begin{enumerate}
\item For every $i = 1, \dots, N$ the strategy $\xi^{\ast, i}  \in \mathcal{A}_{x_i} $ is a minimizer of $J(\cdot; \xi^{\ast, -i})$. Let $X^{\ast, i}$ be the corresponding state process and $\tau^{\ast, i}$ the corresponding absorption time, and $\tau^i:=\tau_\mu(x^i)$. Assume to the contrary that 
$$\Vert \xi^{\ast, i} - \xi^{\mu, i}\Vert_{L^{2}} \neq 0.$$
\begin{enumerate}[label=(\roman*)]
\item $\tau^{i} \ge \tau^{\ast, i}$:
The argument is different from the one given in the proof of Theorem~\ref{thm:N-player-verification} as we are now working with the equilibrium strategy $\xi^{\ast,i}$ and not with the admisible strategy $\xi^{\mu,i}$. Nevertheless, we use again the integration by parts argument based on the dynamics of $\xi^{\mu, i}$. 
Therefore, we need to separate cross terms from squared terms in the convexity estimate:
\begin{align*}
		J(\xi^{\ast, i}; \xi^{\ast,-i})&-J(\xi^{\mu, i}; \xi^{\ast,-i}) \\
		=&~\int_0^{\tau^{i}}\left(	\frac{1}{2}\eta_t(\xi^{\ast,i}_t)^2+\frac{1}{2}\lambda_t (X^{\ast, i}_t)^{2} - \frac{1}{2}\eta_t(\xi^{\mu,i}_t)^2 - \frac{1}{2}\lambda_t (X^{\mu, i}_t)^{2} \right)\,dt \\	
		&+~ \int_0^{\tau^{i}} \frac{\kappa_t}{N} \left( \xi^{\ast, i}_t X^{\ast, i}_t - \xi^{\mu, i}_t X^{\mu, i}_t \right) dt + \int_0^{\tau^{i}}\kappa_t\left( \frac{1}{N}\sum_{j\neq i}\xi^{\ast, j}_t \right)(X^{\ast, i}_t - X^{\mu,i}_t) dt\\
				=&~\int_0^{\tau^{i}}\left(	\frac{1}{2}\eta_t(\xi^{\ast,i}_t)^2+\frac{1}{2}\lambda_t (X^{\ast, i}_t)^{2} - \frac{1}{2}\eta_t(\xi^{\mu,i}_t)^2 - \frac{1}{2}\lambda_t (X^{\mu, i}_t)^{2} \right)\,dt \\	
				&+~ \int_0^{\tau^{i}} \frac{\kappa_t}{N} \left( \xi^{\ast, i}_t  - \xi^{\mu, i}_t \right) X^{\mu, i}_t dt + \int_0^{\tau^{i}}\kappa_t\mu_t(X^{\ast, i}_t - X^{\mu,i}_t) dt\\
			>&~ \int_0^{\tau^{i}}\left(	\eta_t\xi^{\mu, i}_t(\xi^{\ast, i}_t - \xi^{\mu, i}_t) +\lambda_t X^{\mu, i}_t(X^{*,i}_t - X^{\mu, i}_t) \right)\,dt\\
				&+~ \int_0^{\tau^{i}} \frac{\kappa_t}{N} \left( \xi^{\ast, i}_t  - \xi^{\mu, i}_t \right) X^{\mu, i}_t dt + \int_0^{\tau^{i}}\kappa_t\mu_t(X^{\ast, i}_t - X^{\mu,i}_t) dt\\
		=&~ \int_0^{\tau^{i}}\left(	\left(\eta_t\xi^{\mu, i}_t + \frac{\kappa_t}{N}X^{\mu, i}_t\right)(\xi^{\ast,i}_t - \xi^{\mu, i}_t) +\left(\lambda_tX^{\mu, i}_t +\kappa_t\mu_t \right)(X^{\ast, i}_t - X^{\mu, i}_t) \right)\,dt,\\
		=&~0.
\end{align*}
The last argument follows by the same integration by parts argument as in the proof of Theorem~\ref{thm:N-player-verification}.
The above estimate is a contradiction to the optimality of $\xi^{\ast, i}$, hence it must hold $\xi^{\ast, i} = \xi^{\mu,i}$.

\item $\tau^{i} < \tau^{\ast, i}$: Since $\mu$ is either strictly non-negative or strictly non-positive we again only have to consider the case $x_i \mu \ge 0$. 
We then have immediately
\begin{align*}
J(\xi^{\ast, i}; \xi^{\ast,-i})-&J(\xi^{\mu, i}; \xi^{\ast,-i})\\
\geq &~ \int_{\tau^{i}}^{\tau^{\ast, i}} \left(\frac{1}{2}\eta_t(\xi^{\ast, i}_t)^2 + \frac{1}{N}\left(\xi^{\ast, i}_t + \sum_{j\neq i} \xi^{\ast, j}_t \right) \kappa_t X^{\ast, i}_t+\frac{1}{2}\lambda_t(X^{\ast, i}_t)^2\right) \,dt\\
=&~ \int_{\tau^{i}}^{\tau^{\ast, i}} \left(\frac{1}{2}\eta_t(\xi^{\ast, i}_t)^2 + \mu_t \kappa_t X^{\ast, i}_t+\frac{1}{2}\lambda_t(X^{\ast, i}_t)^2\right) \,dt  \\
>&~ 0,
\end{align*}
where the last inequality follows from the fact that $\sign(X^{\ast, i}_t) = \sign(x_i) = \sign(\mu_t)$ for all $t\in[0, \tau^{\ast, i})$. 
The above estimate contradicts the optimality of $\xi^{\ast, i}$ and therefore $\xi^{\ast, i} = \xi^{\mu, i}$.
\end{enumerate}

\item The uniqueness of equilibria with continuous mean trading rate now follows from the first part of this theorem and Theorem~\ref{thm:N-player-verification}  by using the same arguments as in the proof of Theorem~\ref{thm:mfg-final}.
\end{enumerate}
\end{proof}

 \subsection{Convergence to MFG equilibrium}\label{sec:convergence}

In this section we prove the convergence of the $N$-player equilibrium to the mean field equilibrium if the number of players tends to infinity. For each integer $N$ we denote by $\mu^{N}$ the mean equilibrium trading rate in the $N$-player game with initial distribution $\nu^{N}_0=\frac{1}{N}\sum_{j=0}^N\delta_{x_j}$ and tail probability function $q_0^{N}$.

\begin{proposition}
Assume that tail distribution $q^{N}_0$ converges toward $q_0$ in the following sense:
\begin{equation}\label{eq:convergence-assumtion}
\lim_{N\rightarrow\infty}\sup_{x \in [0,\infty)}|q^N_0(x)-q_0(x)|=0.
\end{equation}
Then %
\begin{equation}\label{eq:convergence-equation}
 	\lim_{N\to\infty} \Vert \mu^{N} - \mu \Vert_\infty = 0,
\end{equation}
where $\mu\in C^{0}([0,T];\bR)$ satisfies the mean field equilibrium equation \eqref{eq:consistency}.
\end{proposition}
\begin{proof}
We assume w.l.o.g.~that $\mathbb{E}[\nu_0] >0$ and consider $N$ large enough such that $\mathbb{E}[\nu^{N}_0] > 0$ and %
$\lambda-\frac{\kappa}{N}\geq 0$ and $\eta-\frac{\kappa}{N}>0$.
The equilibrium equation \eqref{eq:mu-second-integral} applied to the empirical measure $\nu_0^{N}$ yields
\begin{align*}
0 <  \mu^{N}_t &\le \frac{\alpha^N_T}{\eta_T}\left\{  \mathbb E[\nu_0^N] - (1-q_0^N(0))x_{\mu^N}  - \int_0^{x_{\mu^N}} q_0^N(y)\,dy			\right\}   + C\int_t^{T}\left(\mu^{N}_s + \int_{s}^T \mu^{N}_u \,du \right)\,ds \\
&\le C \left\{  \mathbb E[\nu_0^N] - (1-q_0^N(0))x_{\mu^N}  - \int_0^{x_{\mu^N}} q_0^N(y)\,dy	\right\} + C\int_t^{T}\mu^{N}_s\, ds,
\end{align*}
where $C>0$ is independent of $N$. %
In view of \eqref{eq:convergence-assumtion} the first integral term in last line above is bounded uniformly in $N$.
Hence it follows from Gr\"onwall's inequality that $(\mu^N)_{N\in \mathbb N}$ is uniformly bounded in $C^{0}([0,T];\bR)$.
Similarly, from \eqref{eq:mu-second-integral} applied to $\nu^N_0$ we have for any $0 \le s \le t \le T$ that
\begin{align*}
0 <  \vert \mu^{N}_t - \mu^{N}_s\vert &\le C(1+T)\left( \sup_{N\ge1}\Vert \mu^{N}\Vert_\infty\right) |t-s|.
\end{align*}
Thus, the sequence $(\mu^N)_{N\in \mathbb N}$ is uniformly Lipschitz continuous.
By the Arzelà–Ascoli theorem there exists $\mu\in C^{0}([0,T];\bR)$ and a subsequence $(N_k)_{k\ge1}$ that converges to $\mu$. 

In what follows we prove that $\mu$ is the unique solution to our integral equation \eqref{eq:mu-second-integral} corresponding to $\nu_0$.   In particular, the sequence $(\mu^N)_{N\in \mathbb N}$ has a unique accumulation point and hence converges to $\mu$. To this end, we  consider the $N$-player equilibrium equation \eqref{eq:mu-second-integral}:
 \begin{align}\label{eq:N-equilibrium-integral}
 \begin{split}
 \mu^{N_k}_t =&~ \frac{\alpha^{1/N_k}_T}{\eta_T} \left\{  \mathbb E[\nu_0^N] - (1-q_0^N(0))x_{\mu^{N_k}}  - \int_0^{x_{\mu^{N_k}}} q_0^{N_k}(y)\,dy	\right\} \\
 &~ +\int_t^{T}\frac{\lambda_s + \frac{\dot\kappa_s}{N_k}}{\eta_s} \left(\int_{s}^T \mu^{N_k}_u \,du\right)\,ds  +\int_t^{T}\left( \frac{\dot\eta_s + \kappa_s(1 - q_0^{N_k}(0))}{\eta_s}-\frac{\kappa_s}{N_k\eta_s}	\right)\mu^{N_k}_s \,ds \\
  &~+\int_t^{T}\frac{\kappa_s}{\eta_s}q_0^{N_k}\left(\int_0^s h^{1/N_k}_u\kappa_u\mu^{N_k}_u\,du\right)\mu^{N_k}_s \,ds,  \qquad t\in[0,T].
\end{split}
\end{align}

Our goal is to take the limit for $k\to \infty$ in the above equation.
The convergence $\alpha^{1/N}_T \to \alpha^0_T$ is proven in Lemma \ref{lem:convergence-alpha}.
By Lemma \ref{lem:about-h} the family $(h^{1/N})_{N\in\bN}$ is uniformly bounded.
By dominated convergence theorem we have that
 \[
 	\lim_{k\to \infty}  x_{\mu^{N_k}}
 	=\lim_{k\to \infty}\int_0^Th^{{1/N_k}}_s\kappa_s\mu^{N_k}_s\,ds
 	=\int_0^Th^{0}_s\kappa_s\mu_s\,ds
 	=  x_\mu. 
 \]
 
Similarly for the terms in the second line of \eqref{eq:N-equilibrium-integral}, it holds by the dominated convergence theorem
 \begin{multline*}
\lim_{k\to\infty}\left(\int_t^{T}\frac{\lambda_s + \frac{\dot\kappa_s}{N_k}}{\eta_s} \left(\int_{s}^T \mu^{N_k}_u \,du\right)\,ds +\int_t^{T}\left( \frac{\dot\eta_s + \kappa_s(1- q^{N_k}_0(0))}{\eta_s}-\frac{\kappa_s}{N_k\eta_s}	\right)\mu^{N_k}_s \,ds\right) \\
=
\int_t^{T}\frac{\lambda_s}{\eta_s} \left(\int_{s}^T \mu_u \,du\right)\,ds +\int_t^{T} \frac{\dot\eta_s + \kappa_s(1- q_0(0))}{\eta_s}\mu_s \,ds, \quad t\in [0,T].
\end{multline*}

Regarding the last term in \eqref{eq:N-equilibrium-integral}, we use a similar estimate as in the proof Theorem \ref{thm:mfg-ode-solution}.
More precisely using integration by parts and the Lipschitz continuity of $Q_0(x) := \int_0^{x} q_{0}(y)\,dy$ and $Q^{N_k}_0(x) := \int_0^{x} q^{N_k}_{0}(y)\,dy$, it follows that for any $t\in(0,T]$ there exists a constant $K>0$ that only depends on the bounds of the coefficients and the functions $h^{\delta}$, $1/h^{\delta}$ and $\dot{h}^{\delta}$ on the interval $[t,T]$, such that
\begin{align*}
&~\left\vert \int_t^{T}\frac{\kappa_s}{\eta_s}q_0^{N_k}\left(x_{\mu^{N_k}}-\int_s^T h^{1/N_k}_u\kappa_u\mu^{N_k}_u\,du\right)\mu^{N_k}_s \,ds - \int_t^{T}\frac{\kappa_s}{\eta_s}q_0\left(x_{\mu}-\int_s^T h^{0}_u\kappa_u\mu_u\,du\right)\mu_s \,ds\right\vert \\
		\leq &~		  \left|  \frac{1}{\eta_Th_T^{1/N_k}}Q_0^{N_k}(x_{\mu^{N_k}})- \frac{1}{\eta_Th_T^{0}}Q_0 (x_{\mu})\right|   \\
		&~+	 \left|\frac{1}{\eta_t h_t^{1/N_k}}Q_0^{N_k}\left( x_{\mu^{N_k}} -\int_t^T h^{1/N_k}_u\kappa_u\mu^{1/N_k}_u\,du	\right) - \frac{1}{\eta_t h_t^{0}}Q_0 \left( x_{\mu} -\int_t^T h^{0}_u\kappa_u\mu_u\,du	\right)      	\right|	 	\\
		&~+   	\int_t^T \left|   	Q_0^{N_k}\left( x_{\mu^{N_k}} -\int_s^T h^{1/N_k}_u\kappa_u\mu^{1/N_k}_u\,du  \right)\frac{ \dot\eta_s h^{1/N_k}_s-\eta_s\dot h^{1/N_k}_s }{(\eta_s h_s^{1/N_k})^2}  \right. \\ & \left. - ~ Q_0 \left( x_{\mu} -\int_s^T h^{0}_u\kappa_u\mu_u\,du  \right)\frac{ \dot\eta_s h^{0}_s-\eta_s\dot h^0_s }{(\eta_s h^0_s)^2}   \right|\,ds\\
\le &~ 
K \left(\sup_{x \in [0,\infty)}|q^{N_k}_0(x)-q_0(x)|\right)\Vert \mu^{N_k} - \mu \Vert_\infty
+ K \vert x_{\mu^{N_k}} - x_\mu \vert 
+ K\Vert \mu^{N_k} - \mu \Vert_\infty.
\end{align*}
Taking $k\to\infty$ in \eqref{eq:N-equilibrium-integral}, it follows from \eqref{eq:convergence-assumtion} 
that $\mu$ satisfies the mean field equilibrium equation \eqref{eq:mu-second-integral}.
\end{proof}

\section{Examples}

In this section we illustrate the impact of market drop-out on equilibrium trading rates in the mean field and $N$-player game on a few numerical examples.
In particular, we compare the equilibria to the mean field game of portfolio liquidation \textit{without} market drop-out as presented in \cite{FGHP-2018}.
For simplicity we will consider constant cost parameters; more precisely we choose 
\begin{align*}
\eta \equiv 5, \quad \kappa \equiv 10, \quad  \lambda \equiv 5.
\end{align*}

\begin{figure}
\begin{center}
\includegraphics[scale=.7]{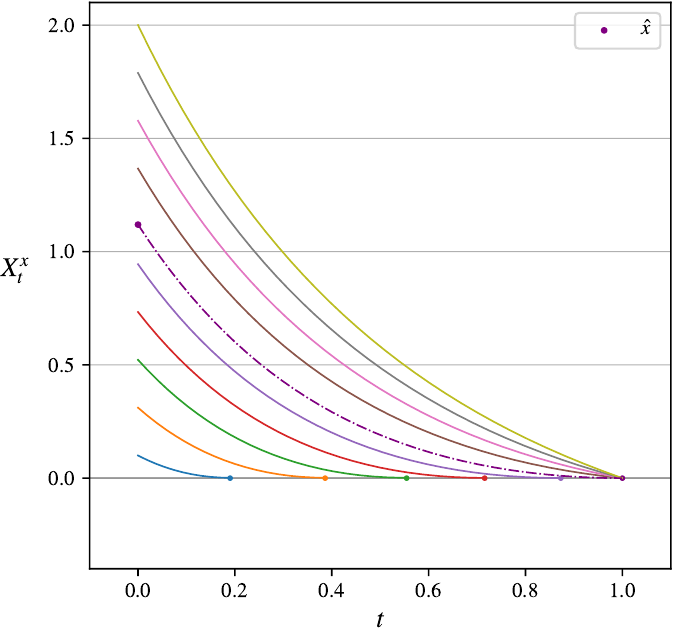}
\hfill{}\includegraphics[scale=.7, trim=.8cm 0 0 0,clip]{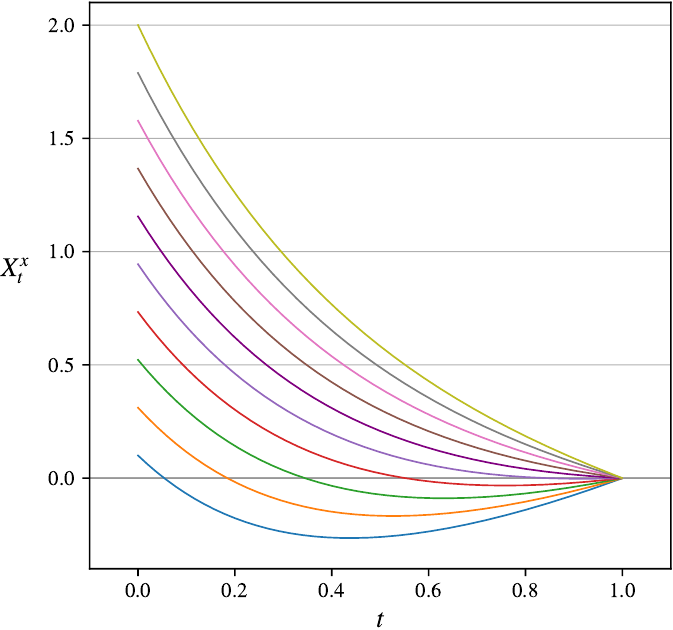}
\end{center}
\caption{The evolution of the state processes in equilibrium of the liquidation game \emph{with market drop-out} (left) and \emph{without market drop-out} (right) for several representative players in the presence of only sellers. We have highlighted the moments of drop-out and $\hat{x}$ (which represents the smallest initial position for which no early exit takes place).}\label{fig:one-sided-rep}
\end{figure}

A numerical approximation of the mean-field equilibrium is obtained by solving the backward integral equation \eqref{eq:mu-third-integral} using a standard numerical solver for several values c\footnote{Technically one rather defines a function that outputs for any given parameter $c$ the evaluated solution $f_{\mu^{c}}(T) - c$ and then applies a standard numerical root finding method to that function.} and then finding the root of\footnote{Note that in the two sided situation we use the root of $c\mapsto \mu_T^{c}$ as an upper bound in the root finding procedure. In the one sided situation we instead  use the upper bound of support of $\nu_0$ or the upper bound from the a priori estimate \eqref{eq:mu-sign}.} $c\mapsto f^{c}_0$.

We first consider a one-sided situation with only sellers.
Specifically we choose $\nu_0$ to be an exponential distribution with mean $1.5$, i.e. $q_0(x) = e^{-\frac{2}{3}x}$ and $p_0\equiv 0$.
Figure~\ref{fig:one-sided-rep} presents the evolution of the equilibrium mean-field state process with and without market drop-out for several representative players. In a model without market drop-out constraint,  players with smaller initial portfolio use short selling and round trip strategies, whereas in the presence of the market drop-out constraint all players stick to strict selling strategies (as it is generally the case for the dominating side of the market, see Section~\ref{sec:single_player_verification}).

\begin{figure}
\begin{center}
\includegraphics[scale=.7]{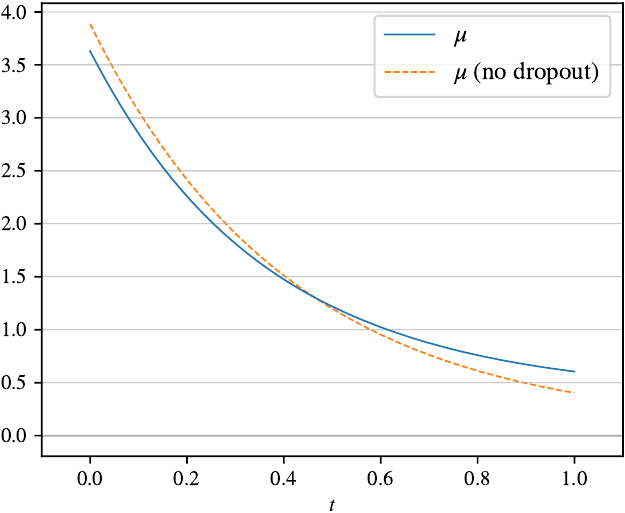}
\hfill{}\includegraphics[scale=.7]{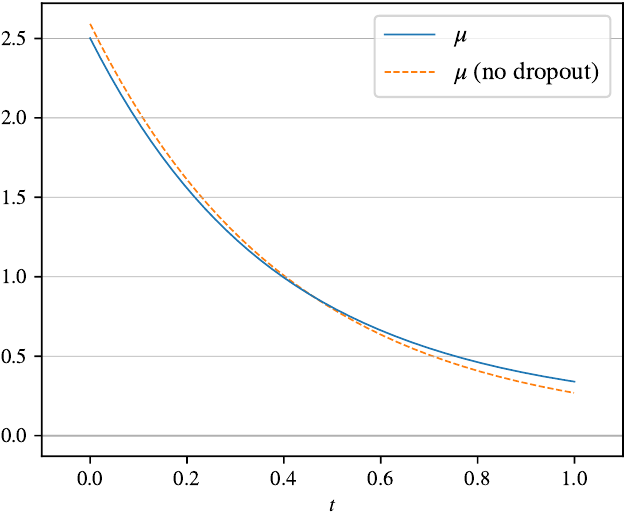}
\end{center}
\caption{Comparison of the mean trading rate in the equilibrium of the liquidation game \emph{with market drop-out} (solid line) and \emph{without market drop-out} (dashed line). \textit{left:} one-sided situation,  \textit{right:} two sided situation.}\label{fig:one-sided-mean}
\end{figure}

Figure~\ref{fig:one-sided-mean} (left) compares the mean equilibrium trading rate for both situations.
For our choice of parameters the rates only deviate slightly. As expected, the drop-out constraint initially leads to slower {aggregate} liquidation and to an increasingly faster {aggregate} liquidation halfway through the trading period  {when compared to model without constraints. The reason is that some players would liquidate fast and then change the direction if they were allowed to; these players now trade slower initially. The aggregate rate is faster towards the end of the trading period as there are now no buyers in the market. As a result, in our current setting the asset price initially decreases faster and then slower compared to the unconstrained case. However, in both settings $\int_0^T\mu_t\, dt = \mathbb{E}[\nu_0]$. Hence if the permanent impact parameter is constant, then the price at the terminal time is the same in both settings.} 

Secondly, we consider a two-sided situation with $q_0(x) = 0.8 \cdot e^{- \frac{2}{3} x}$ and  $p_0(x) = 0.2 \cdot e^{x}$. This leads to a mean initial position of $\E[\nu_0] = 1$, hence, to a situation in which sellers dominate the market.
Figure~\ref{fig:two-sided-rep} (left) shows the evolution of the representative state processes while Figure~\ref{fig:two-sided-rep} (right) displays the equilibrium mean trading rate.
We observe that the equilibrium term indeed does not change its positive sign, which again leads to the fact that sellers do not use short selling strategies in the situation with market drop-out.

Figure~\ref{fig:n-player} shows the solution to the $N$-player games. 
In this game, the two smallest players on the selling side drop out early, while all other players liquidate at the terminal time. Our simulations suggest that the convergence to the MFG equilibrium is rather fast. We emphasize that mean trading rate in the $N$-player game is not smooth; the slope changes discontinuously when a player drops out of the market.

\begin{figure}
\begin{center}
\includegraphics[scale=.65]{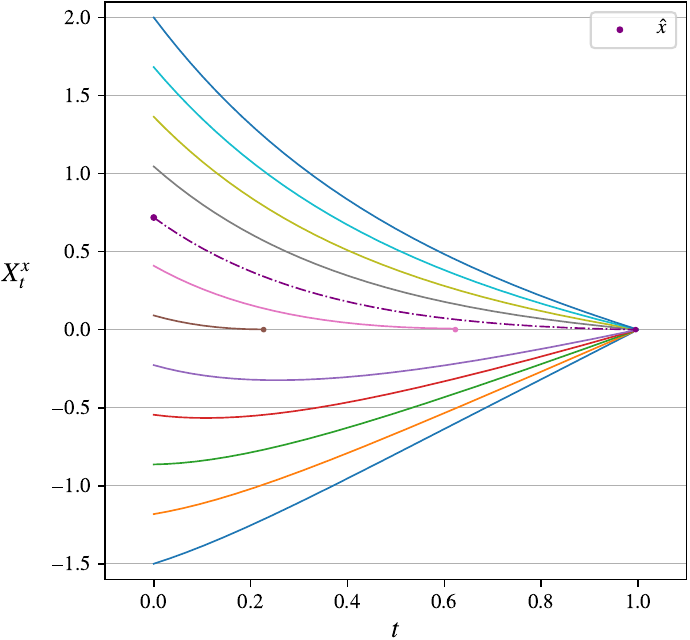}
\hfill{}\includegraphics[scale=.65, trim=0.9cm 0 0 0,clip]{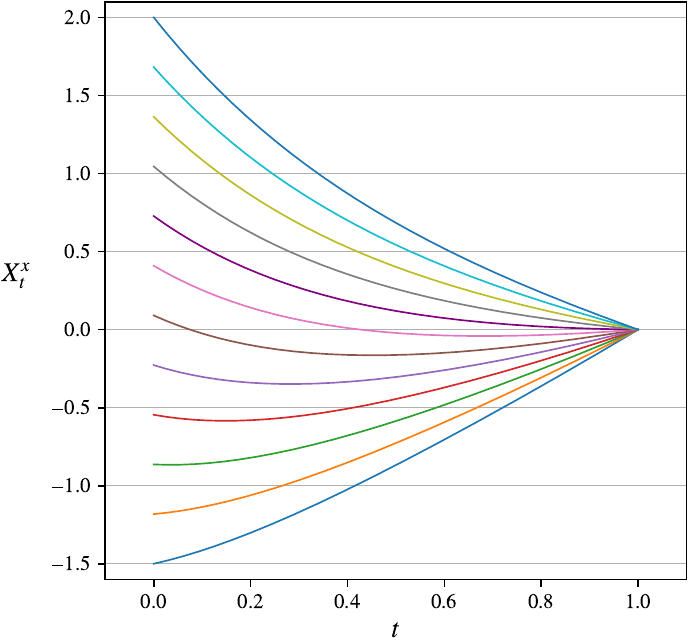}
\end{center}
\caption{The evolution of the state processes in equilibrium of the liquidation game in the two sided case: \emph{with market drop-out} (left) and \emph{without market drop-out} (right).}\label{fig:two-sided-rep}
\end{figure}

\begin{figure}[h!]
\begin{minipage}{\linewidth}
  \centering
  $\vcenter{\hbox{
  \includegraphics[scale=.65]{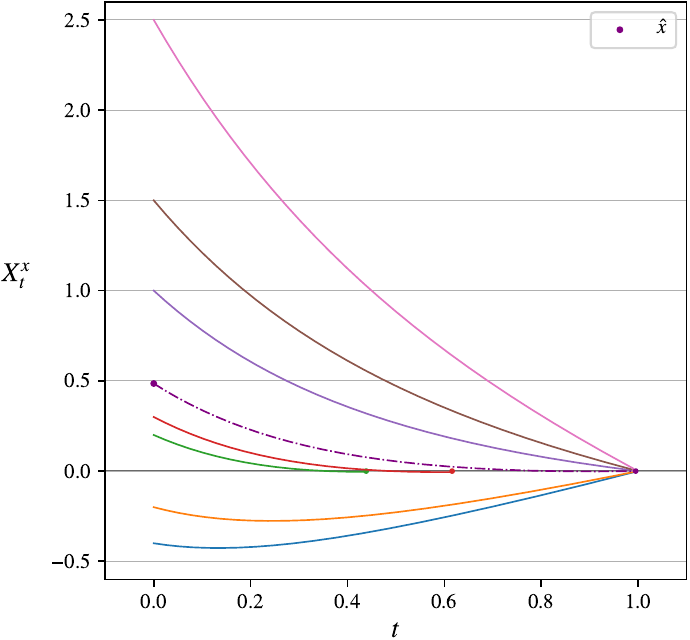}
  }}$
  \hspace*{1em}
  $\vcenter{\hbox{
  \includegraphics[scale=.65]{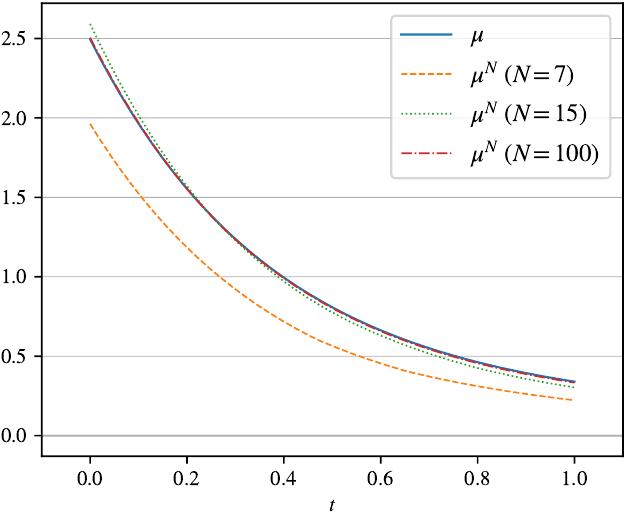}  
  }}$
\end{minipage}
\caption{\textit{left}: Evolution of states of $N=7$ players in the Nash equilibrium.
\textit{right}: Comparison of the mean trading rate in the equilibrium of the \emph{mean field game} (solid line) and the $N$-player game (dashed/dotted line) for $N=7, 15$ and $100$.}\label{fig:n-player}
\end{figure}

\section{Conclusion and outlook}
In this paper we studied novel mean field liquidation games, where each player drops out of the market once her position hits zero. We characterized the unique Nash equilibrium by a nonlinear integral equation. We established existence and uniqueness of equilibrium result for both the $N$-player game and the MFG, and proved the convergence of equilibria to the MFG equilibrium. Our result suggests that in a market dominated by sellers, the sellers' optimization problem with market drop-out is the same as the optimization problem with short selling constraint. However, the buyers with small initial positions would trade in an opposite direction initially, thus, the buyers' optimization problem may be different from the one with short selling constraint, which is left as an independent ongoing work. 

{Several avenues are open for further research. For instance, one of the main limitations of our model is that all players share the same trading horizon. However, we expect that our model provides a mathematical framework for analyzing more general models. To illustrate this, let us consider a model with two risk-neutral players where $\kappa$ is constant and where player $i=1,2$ holds an initial position $x_i$ and needs to liquidate this position by time $T_i$. Let $T_1 < T_2$. To solve this liquidation game we can proceed as follows. First, we can solve a game with common trading horizon $T=T_2$ using our results. If player 1 liquidates before time $T_1$, the game with different trading horizons is solved. If not, then we can consider a family of games with common trading horizon $T=T_1$ indexed by the position $x'_2$ that player $2$ needs to hold at time $T_1$. Since both players are risk neutral, this is equivalent to a 2-player liquidation game on $[0,T_1]$ where player 2 holds an initial position $x_2-x_2'$. This game can again be solved using the results of this paper. This player's total liquidation cost can then be computed by first solving the new 2-player game on $[0,T_1]$ and then solving a single player liquidation problem on $[T_1, T_2]$ with initial position $x_2'$. In a final step one then needs minimize player 2's total cost over $x_2' \in [0,x_2]$. Of course, in an MFG the situation is much more involved.}

\appendix
\titleformat{\section}[hang]{\normalfont\bfseries\Large}{}{0pt}{}
\section{Appendix}

Throughout this section we assume that Assumption \ref{ass:single-player} holds.
\begin{lemma}[cf.  Lemma A.1 in \cite{FGHP-2018}]\label{lem:estimatesA}
There exists a unique solution to the terminal value problem
\begin{align}\label{eq:A-appendix}
-\dot A_t=&~-\frac{A^2_t}{\eta_t}+\lambda_t, \qquad \lim_{t\nearrow T}A_t=\infty.
\end{align}
Moreover,  the following estimates hold for all $0 \le s \le t < T$:
\begin{align*}
\frac{1}{\etamin(T-t)} \;\;\le\;\; \frac{1}{\int_{t}^{T}\frac{1}{\eta_s} ds} \;\;\le\;\;A_t \;\;\le\;\; \frac{1}{(T-t)^{2}}\int_t^{T}\left(\eta_s +(T-s)^{2}\lambda_s \right)\,ds
\end{align*}
and
\begin{align*}
 \exp\left(-\int_{s}^{t}\frac{A_r}{\eta_r}dr\right) \le  {\etamax}{\etamin}\frac{T-t}{T-s}.
\end{align*}

\end{lemma}

\begin{lemma}\label{lem:alpha-exisitence} Let $A$ be the solution to the equation \eqref{eq:A-appendix}.
Then the map
$$\alpha^{0}: [0,T) \to [0,\infty),\quad t \mapsto A_t e^{-\int_0^t\frac{A_r}{\eta_r}\,dr},$$
is non-increasing and 
 $$\alpha^{0}_T = \lim_{t \to T} A_t e^{-\int_{0}^{t}\frac{A_r}{\eta_r}dr} \in (0,\infty).$$
\end{lemma}
\begin{proof}
Differentiation yields
\begin{align*}
\frac{d}{dt}\left(A_t e^{-\int_0^t\frac{A_r}{\eta_r}\,dr}\right)
=\left(\dot A_t -\frac{A^2_t}{\eta_t}\right)e^{-\int_0^t\frac{A_r}{\eta_r}\,dr}
=- \lambda_t e^{-\int_0^t\frac{A_r}{\eta_r}\,dr} \le 0, \qquad t\in[0,T),
\end{align*}
which proves the first statement.
Since $\alpha^{0}$ is positive it remains to show that $\alpha^{0}_T$ is strictly positive.
The differentiability of $\eta$ at $T$ and along with the fact that $\eta_T>0$ yields $L>0$ and $\varepsilon\in(0,\frac{\eta_T}{L})$ such that
$$ \vert \eta_T - \eta_t \vert \le L(T-t), $$
for all $t \in (T-\varepsilon, T]$. 
For all such $t$ it then follows from Lemma \ref{lem:estimatesA} that
\begin{align*}
\int_{T-\epsilon}^t\frac{A_r}{\eta_r}\,dr 
&\le~ \int_{T - \varepsilon}^t \frac{\frac{1}{(T - s)^2} \int_s^T \eta_u \,du +
T \Vert \lambda \Vert_{\infty}}{\eta_s} \, ds \\
&\le~ \int_{T - \varepsilon}^t \frac{\eta_T \frac{1}{(T - s)} + \frac{1}{2} L + T
\Vert \lambda \Vert_{\infty}}{\eta_T - L (T - s)} \,ds \\
&=~ \int_{T - t}^{\varepsilon} \frac{1}{u - \frac{L}{\eta_T} u^2} \,d u +
\left( \frac{1}{2}  + \frac{T \Vert \lambda \Vert_{\infty}}{L} \right) \int_{T - t}^{\varepsilon}
\frac{1}{\frac{\eta_T}{L} - u} \,du.
\end{align*}
For the two integrals on the right we then have
\begin{align*}
  \int_{T - t}^{\varepsilon} \frac{1}{u - \frac{L}{\eta_T} u^2} \,d u 
  = -\ln \left( \frac{T - t}{\varepsilon} \right)+\ln \left( \frac{\frac{\eta_T}{L} - (T - t)}{\frac{\eta_T}{L} -
  \varepsilon} \right)
  \le -\ln \left( \frac{T - t}{\varepsilon} \right) - \ln \left( 1
  - \frac{L}{\eta_T } \varepsilon \right)
\end{align*}
and similarly
\begin{align*}
\int_{T - t}^{\varepsilon}
\frac{1}{\frac{\eta_T}{L} - u} \,d u = \ln \left( \frac{\frac{\eta_T}{L} - (T - t)}{\frac{\eta_T}{L} -
  \varepsilon} \right) \le - \ln \left( 1
  - \frac{L}{\eta_T } \varepsilon \right).
\end{align*}
Putting together the estimates yields 
\begin{align*}
e^{- \int_{T - \varepsilon}^t \frac{A_s}{\eta_s} \,d s} \geqslant
   K_{\varepsilon}\varepsilon^{- 1} (T - t), \quad \text{with } K_{\varepsilon} =  \left( 1 - \frac{L}{\eta_T }
\varepsilon \right)^{   \frac{3}{2}  + \frac{T \Vert \lambda \Vert_{\infty}}{L} } \in (0,1),
\end{align*}
for all $t \in (T-\varepsilon, T]$.
This allows us to finally estimate using again Lemma \ref{lem:estimatesA}
\begin{align*}
 \alpha^{0}_T =
  e^{- \int_0^{T - \varepsilon} \frac{A_s}{\eta_s} \,d s}\left( \lim_{t
  \rightarrow T} A_t e^{- \int_{T - \varepsilon}^t \frac{A_s}{\eta_s} \,d
  s}\right)
  \ge e^{- \int_0^{T - \varepsilon} \frac{A_s}{\eta_s} \,d s}\varepsilon^{-1}
  K_{\varepsilon} \etamin^{-1} >0.
\end{align*}
\end{proof}

\begin{lemma}\label{lem:A-kappa}
For any $\delta \in [0,1]$ there exists a unique solution $A^{\delta}$ to the singular terminal value problem
\begin{align}\label{eq:A-appendix-kappa}
-\dot A_t=&~-\frac{1}{\eta_t}A^2_t + \delta\frac{\kappa_t}{\eta_t}A_t+\lambda_t, \qquad \lim_{t\nearrow T}A_t=\infty.
\end{align}
Furthermore, the solutions satisfy the  comparison principle
$$A^{\delta_1} \le A^{\delta_2}, \qquad 0 \le \delta_1 \le \delta_2 \le 1$$
and the map
\begin{equation}\label{eq:A-comparison}
[0,1] \mapsto (C^{0}(I; \bR), \Vert\cdot\Vert_\infty), \quad \delta \mapsto A^{\delta},
\end{equation}
is continuous uniformly on compact subintervals $I\subset[0,T)$.
\end{lemma}
\begin{proof}
For $\delta = 0$ the equation \eqref{eq:A-appendix-kappa} coincides with \eqref{eq:A-appendix} and we denote $A = A^{0}$.
For $\delta\in(0,1]$ define
\begin{align}\label{eq:A_modified_coefficients}
\widetilde{\eta}_t := \eta_t e^{\int_0^{t}\delta\frac{\kappa_s}{\eta_s} ds}, \qquad
\widetilde{\lambda}_t := \lambda_t e^{\int_0^{t}\delta\frac{\kappa_s}{\eta_s} ds}, \qquad t\in[0,T].
\end{align}
Then clearly $\widetilde{\eta} \in C^{1}([0,T], (0,\infty))$ and $\widetilde{\lambda}$ is bounded and therefore by the above Lemma \ref{lem:estimatesA} there exists a unique solution $\tilde{A}$ to the singular terminal value problem
\begin{align*}
-\frac{d}{dt}\widetilde{A}_t=&~-\frac{1}{\widetilde\eta_t}\widetilde{A}^2_t+\tilde\lambda_t, \qquad \lim_{t\nearrow T}\widetilde{A}_t=\infty.
\end{align*}
A simple calculation then shows that $A^{\delta} := \tilde A e^{-\int_{0}^{\cdot}\delta\frac{\kappa_r}{\eta_r}dr}$ solves \eqref{eq:A-appendix-kappa}.
{Conversely, for any solution $A^{\delta}$ to \eqref{eq:A-appendix-kappa} we can define $\widetilde{A}:= A^{\delta}e^{\delta\int_0^{\cdot} \frac{\kappa_r}{\eta_r}dr}$ which then solves \eqref{eq:A-appendix} with coefficients $\widetilde{\eta}$ and $\widetilde{\lambda}$. 
By uniqueness of solutions to \eqref{eq:A-appendix} it thus follows that also \eqref{eq:A-appendix-kappa} has a unique solution.}
For the forthcoming part of the proof, note that $A^{\delta} > c \tilde{A}>0$ for a suitable constant $c>0$.

Now let $0 \le \delta_1 \le \delta_2 \le 1$ and let $A^{\delta_2}$ and $A^{\delta_2}$ be the corresponding solutions to \eqref{eq:A-appendix-kappa}. Then, 
	\begin{equation}\label{eq:difference-A-AN}
		\left(\frac{1}{A^{\delta_1}_t}-\frac{1}{A^{\delta_2}_t}\right)' = \frac{\kappa_t}{\eta_t}\left(\frac{\delta_1}{A^{\delta_1}_t} - \frac{\delta_2}{A^{\delta_2}_t}\right) + \lambda_t\left(\frac{1}{A^{\delta_1}_t} + \frac{1}{A^{\delta_2}_t}\right)\left(\frac{1}{A^{\delta_1}_t} - \frac{1}{A^{\delta_2}_t}\right).
	\end{equation}
Since it follows from the first part of the proof that $A^{\delta}>0$ we can estimate
	\begin{equation*}
		\frac{d}{dt}\left(\frac{1}{A^{\delta_1}_t}-\frac{1}{A^{\delta_2}_t}\right) \le \left(\frac{\delta_2\kappa_t }{\eta_t}+\lambda_t\left(\frac{1}{A^{\delta_1}_t} + \frac{1}{A^{\delta_2}_t}\right)\right)\left(\frac{1}{A^{\delta_1}_t} - \frac{1}{A^{\delta_2}_t}\right).
	\end{equation*}
Recalling that $\lim_{t\to T} \frac{1}{A^{\delta}_t}=0$ and applying the Gr\"onwall's inequality in the backward form we obtain
$$\frac{1}{A^{\delta_1}_t}-\frac{1}{A^{\delta_2}_t} \ge 0, \quad t \in [0,T],$$
which proves the comparison principle.
This being proven, we can estimate in \eqref{eq:difference-A-AN} as follows:
	\begin{align*}
		\frac{d}{dt}\left(\frac{1}{A^{\delta_1}_t}-\frac{1}{A^{\delta_2}_t}\right) ~\ge~ -\frac{\kappa_t}{\eta_t}
		\frac{\delta_2-\delta_1}{A^{\delta_2}_t} 
		~\ge~ -\frac{\kappa_t}{\eta_t}
		\frac{\delta_2-\delta_1}{A_t}.
	\end{align*}
Integrating both sides of the inequality and using Lemma \ref{lem:estimatesA} yields
		\begin{equation*}
		\frac{1}{A^{\delta_1}_t}-\frac{1}{A^{\delta_2}_t} 
		~\le~ (\delta_2-\delta_1)\int_{t}^{T}\frac{\kappa_s}{\eta_s}
		\frac{1}{A_s}\, ds 
		~\le~ C(\delta_2-\delta_1), \qquad t\in[0,T)
	\end{equation*}
for some constant $C>0$. From the comparison principle we then obtain
		\begin{equation*}
		0 \le A^{\delta_2}_t - A^{\delta_1}_t \le (A^1_t)^{2} C (\delta_2 - \delta_1), \qquad t \in [0, T).
	\end{equation*}
The continuity of the map $\delta \mapsto A^{\delta}$ now follows from the uniform boundedness of the map $A^{1}$ on any compact subset of $[0,T)$.  
\end{proof}

\begin{lemma}\label{lem:convergence-alpha}
For all $\delta \in [0,1]$ the map $\alpha^\delta: t\mapsto A^\delta_t e^{-\int_0^t \frac{A_r^\delta-\delta\kappa_r}{\eta_r}\,dr}$ is non-increasing and
\begin{equation}\label{eq:alphaN-appendix}
\alpha^{\delta}_T = \lim_{t \to T} A^{\delta}_t e^{-\int_{0}^{t}\frac{A^{\delta}_r - \delta\kappa_r}{\eta_r}dr} \in (0,\infty).
\end{equation}
Furthermore we have the following convergence
\begin{equation*}
\lim_{\delta \to 0}\alpha^\delta_T=\alpha^{0}_T.
\end{equation*}
\end{lemma}
\begin{proof}
	It can be verified directly that 
\[
	\frac{d}{dt}\left(A^\delta_t e^{-\int_0^t \frac{A_r^\delta-\delta\kappa_r}{\eta_r}\,dr}\right)=-\lambda_t e^{-\int_0^t \frac{A_r^\delta-\delta\kappa_r}{\eta_r}\,dr}\leq 0.
\]		
	 From the proof of Lemma~\ref{lem:A-kappa} we have that $A^{\delta}e^{\int_0^{\cdot}\delta\frac{\kappa_r}{\eta_r}dr}$ satisfies the equation \eqref{eq:A-appendix-kappa} with $\delta = 0$ and modified coefficients $\widetilde{\eta}$ and $\widetilde{\lambda}$ defined in \eqref{eq:A_modified_coefficients} that satisfy Assumption~\ref{ass:single-player}.
Therefore the existence and positivity of the limit in \eqref{eq:alphaN-appendix} follows immediately from Lemma~\ref{lem:alpha-exisitence}.

Regarding the convergence result, note that on the one hand 
	\[
		\alpha^\delta_T=\lim_{t \to T}A^\delta_t e^{-\int_0^t	\frac{	A_s^\delta-\delta\kappa_s	}{\eta_s}	\,ds}=A^\delta_0-\int_0^T   \lambda_s e^{-\int_0^s  \frac{ A^\delta_r-\delta\kappa_r		}{\eta_r}\,dr  }  \,ds.
	\]
From Lemma \ref{lem:A-kappa} it follows that $A^{\delta}$ converges pointwise on $[0,T)$ towards $A$. 
By the dominated convergence theorem we then have
	\begin{equation*}
		\lim_{\delta\to0}\alpha^\delta_T=\lim_{\delta\to0} \left(A^\delta_0-\int_0^T   \lambda_s e^{-\int_0^s  \frac{ A^\delta_r-\delta\kappa_r}{\eta_r}\,dr  }  \,ds\right)=A^{0}_0-\int_0^T\lambda_s e^{-\int_0^s\frac{A^{0}_r}{\eta_r}\,dr}\,ds. 
	\end{equation*}
The desired convergence result now follows from
	\begin{equation*}
		\alpha^{0}_T = \lim_{t\to T}A^{0}_t e^{-\int_0^t\frac{A^{0}_s}{\eta_s}\,ds} 
		= A^{0}_0-\int_0^T\lambda_s e^{	-\int_0^s \frac{A^{0}_r}{\eta_r}\,dr	}\,ds.
	\end{equation*}
\end{proof}

\begin{proof}[Proof of Lemma \ref{lem:about-h}]
From Lemma \ref{lem:convergence-alpha}, we have for any $\delta\in[0,1]$
\begin{equation}\label{eq:A_h_upper}\begin{split}
	h^{\delta}_t = &~   e^{-\int_0^t \frac{A^\delta_r}{\eta_r}\,dr  }	\int_0^t  \frac{1}{\eta_s} e^{\int_0^s \frac{2A^\delta_r-\delta\kappa_r}{\eta_r}\,dr  }\,ds	\\
	=&~ e^{	-\int_0^t\frac{A^{\delta}_r}{\eta_r}\,dr	}\int_0^t  e^{\int_0^s\frac{A^{\delta}_r}{\eta_r}\,dr	}\frac{A^{\delta}_s}{\eta_s} \frac{1}{\alpha^{\delta}_s} \,ds \\	
	\le&~\frac{1}{\alpha^{\delta}_t} e^{	-\int_0^t\frac{A^{\delta}_r}{\eta_r}\,dr	}\left(e^{\int_0^t\frac{A^{\delta}_r}{\eta_r}\,dr	}-1 \right) \\
	 \le&~ \frac{1}{\alpha^{\delta}_T}, \qquad t\in[0,T).\end{split}
\end{equation}
From Lemma~\ref{lem:convergence-alpha}, we have the uniform boundedness of $\left(\frac{1}{\alpha^\delta_T}\right)_{0\leq\delta\leq 1}$, which leads to a uniform upper bound of $h^{\delta}$ for $\delta \in[0,1]$.
On the other hand, note that
\begin{align*}
\frac{d}{dt}\left(\frac{e^{\int_0^t\frac{A^{\delta}_r - \delta \kappa_r}{\eta_r}\,dr}}{A^{\delta}_t} \right)
&=~ e^{\int_0^t\frac{A^{\delta}_r - \delta \kappa_r}{\eta_r}\,dr}\frac{\lambda_t}{(A^{\delta}_t)^{2}}, \qquad t\in[0,T).
\end{align*}
Therefore using integration by parts we can establish the following lower bound on $h^{\delta}$
\begin{equation}\label{eq:A_h_lower}\begin{split}
h^{\delta}_t 
&= e^{-\int_0^t\frac{A^{\delta}_r}{\eta_r}\,dr}\int_0^t \frac{1}{\eta_s} e^{\int_0^s\frac{2A^{\delta}_r - \delta \kappa_r}{\eta_r}\,dr}\,ds \\
&= e^{-\int_0^t\frac{A^{\delta}_r}{\eta_r}\,dr}\int_0^t \left(\frac{A^{\delta}_s}{\eta_s} e^{\int_0^s\frac{A^{\delta}_r}{\eta_r}\,dr}\right) \frac{e^{\int_0^s\frac{A^{\delta}_r - \delta \kappa_r}{\eta_r}\,dr}}{A^{\delta}_s}\,ds \\
&= e^{-\int_0^t\frac{A^{\delta}_r}{\eta_r}\,dr}\left(e^{\int_0^t\frac{A^{\delta}_r}{\eta_r}\,dr}\frac{e^{\int_0^t\frac{A^{\delta}_r  - \delta \kappa_r}{\eta_r}\,dr}}{A^{\delta}_t} - \frac{1}{A^{\delta}_0} - \int_0^t e^{\int_0^s\frac{A^{\delta}_r}{\eta_r}\,dr} \frac{\lambda_se^{\int_0^s\frac{A^{\delta}_r-\delta \kappa_r}{\eta_r}\,dr}}{(A^{\delta}_s)^{2}}\,ds \right) \\
&= \frac{e^{\int_0^t\frac{A^{\delta}_r  - \delta \kappa_r}{\eta_r}\,dr}}{A^{\delta}_t} + e^{-\int_0^t\frac{A^{\delta}_r}{\eta_r}\,dr}\left(-\int_0^t \frac{\lambda_s e^{\int_0^{s}\delta\frac{\kappa_r}{\eta_r}\,dr}}{(\alpha^{\delta}_s)^{2}}\,ds - \frac{1}{A^{\delta}_0}\right)\\
&\ge \frac{1}{\alpha^{\delta}_t} - Ce^{-\int_0^t\frac{A^{\delta}_r}{\eta_r}\,dr}, \qquad t\in[0,T),
\end{split}
\end{equation}
where $C>0$ is a constant that by Lemma~\ref{lem:A-kappa} and Lemma~\ref{lem:convergence-alpha} can be chosen independently of $\delta\in[0,1]$.
Furthermore combining the estimates \eqref{eq:A_h_upper} and \eqref{eq:A_h_lower} and letting $t\nearrow T$ we see that $h^{\delta}_T = 1/\alpha^{\delta}_T$.
Finally, note that $h^{\delta}$ is continuously differentiable on $[0,T)$.
Hence, using the lower bound \eqref{eq:A_h_lower} and once more that $\alpha^{\delta}$ is non-increasing we obtain the following uniform bound over $\delta\in[0,1]$
\begin{align*}
\dot h^{\delta}_t &=~ -\frac{A^{\delta}_t}{\eta_t}h^{\delta}_t + \frac{1}{\eta_t}e^{\int_0^{t}\frac{A^{\delta}_r- \delta\kappa_r}{\eta_r}\,dr} \\
&\le~ -\frac{A^{\delta}_t}{\eta_t}\left(\frac{1}{\alpha^{\delta}_t} - Ce^{-\int_0^t\frac{A^{\delta}_r}{\eta_r}\,dr}
\right) + \frac{1}{\eta_t}e^{\int_0^{t}\frac{A^{\delta}_r- \delta\kappa_r}{\eta_r}\,dr} \\
&=~  -\frac{1}{\eta_t}\frac{A^{\delta}_t}{\alpha^{\delta}_t}+ C\frac{\alpha_t}{\eta_t}e^{ -\int_0^t \frac{\kappa_r}{\eta_r}\,dr }  + \frac{1}{\eta_t}e^{\int_0^{t}\frac{A^{\delta}_r- \delta\kappa_r}{\eta_r}\,dr} \\
&\le~ CA^{\delta}_0\frac{1}{\eta_t} \\
&\le~ CA^{\delta}_0\Vert 1/\eta\Vert_\infty, \qquad t\in[0,T].
\end{align*}
On the other hand, using the upper bound \eqref{eq:A_h_upper} we obtain
\begin{align*}
\dot h^{\delta}_t &\ge~ -\frac{A^{\delta}_t}{\eta_t} \frac{1}{\alpha^{\delta}_t} \left(1-e^{	-\int_0^t\frac{A^{\delta}_r}{\eta_r}\,dr	}\right) + \frac{1}{\eta_t}e^{\int_0^{t}\frac{A^{\delta}_r- \delta\kappa_r}{\eta_r}\,dr} \\
&=~ -\frac{A^{\delta}_t}{\eta_t} \frac{1}{\alpha^{\delta}_t}  + \frac{A^{\delta}_t}{\eta_t} \frac{1}{\alpha^\delta_t}e^{	-\int_0^t\frac{A^{\delta}_r}{\eta_r}\,dr	} + \frac{1}{\eta_t}e^{\int_0^{t}\frac{A^{\delta}_r- \delta\kappa_r}{\eta_r}\,dr} \\
&=~ \frac{1}{\eta_t} e^{-\int_0^t \frac{\delta\kappa_r}{\eta_r}\,dr} \\
&\ge~ \frac{1}{\Vert \eta \Vert_\infty}e^{-\int_0^T \frac{\delta\kappa_r}{\eta_r}\,dr}, \qquad t\in[0,T].
\end{align*}
This proves in particular that $h^{\delta}$ is strictly increasing.

Finally, the a.e. convergence $h^\delta\rightarrow h^0$ and $\dot h^\delta\rightarrow \dot h^0$ can be easily obtained by the expressions of $h^\delta$ and $\dot h^\delta$, as well as the convergence of $A^\delta\rightarrow A^0$ established in Lemma \ref{lem:A-kappa}. 
\end{proof}

The following version of Gr\"onwall's inequality is frequently used in the main text, but it is difficult to locate it in the literature, where
 the function $\beta$ is usually assumed to be non-negative.
\begin{lemma}[Lower Gr\"onwall's inequality]\label{lem:lower_gronwall}
Let $u$ be an absolutely continuous function on $[0,T]$, with $u_T > 0$ and such that there exists a continuous function $\beta$ such that
\begin{align*}
\dot u_t ~\le~ \beta_t u_t  \qquad \text{ for a.e.   }t \in [0,T].
\end{align*}
Then it holds
\begin{align*}
 u_t \ge  u_T\exp\!\left(-\int_t^{T}\beta_s\, ds\right), \qquad t\in[0,T].
\end{align*}
\end{lemma}
\begin{proof}
Define the differentiable function $\nu_\cdot := \exp(\int_\cdot^{T}\beta_sds)/ u_T$. 
Since $u_T>0$ by assumption it holds $\nu >0$.
At any point $t\in[0,T]$ where $u$ is differentiable we obtain
\begin{align*}
\frac{d}{dt}(\nu_t u_t) =&~ \dot\nu_t u_t + \nu_t\dot u_t \\
\le&~ -\beta_t \nu_t u_t + \beta_t \nu_t u_t = 0.
\end{align*}
Since $u$ is almost everywhere differentiable, integration then yields
\begin{align*}
u_T\nu_T - u_t\nu_t = \int_t^{T}\frac{d}{dt}(\nu_su_s)\, ds\le 0,
\end{align*}
and with $\nu_Tu_T = 1$ the stated estimate follows.
\end{proof} 
 
\bibliography{bib_FHH}

\end{document}